\documentclass[reqno]{amsart}
\usepackage{amsthm}
\usepackage{amsmath}
\usepackage{amssymb}
\usepackage{mathtools}   
\usepackage{booktabs}
\usepackage{multirow}
\usepackage{url}
\usepackage{tikz}
\usepackage{accents}     

\usepackage[ruled,vlined,linesnumbered,norelsize,algosection]{algorithm2e}
\SetAlgoHangIndent{0pt}

\usepackage[
   setpagesize=false,
   pagebackref,
   pdfpagelabels,
   plainpages=false,
   pdfstartview={FitH 1000},
   bookmarksnumbered=false,
   linktoc=all,
   colorlinks=true,
   anchorcolor=black,
   menucolor=black,
   runcolor=black,
   filecolor=black,
   linkcolor=black,
   citecolor=black,
   urlcolor=black,
]{hyperref}

\newtheorem{thm}{Theorem}
\newtheorem{prop}[thm]{Proposition}
\newtheorem{lem}[thm]{Lemma}
\newtheorem{cor}[thm]{Corollary}
\theoremstyle{definition}
\theoremstyle{remark}
\newtheorem{rem}[thm]{Remark}

\numberwithin{equation}{section}
\numberwithin{thm}{section}

\newcommand{\ZZ}{\mathbf{Z}}
\newcommand{\CC}{\mathbf{C}}
\newcommand{\RR}{\mathbf{R}}
\newcommand{\norm}[1]{\Vert #1 \Vert}
\newcommand{\bignorm}[1]{\bigg\Vert #1 \bigg\Vert}
\newcommand{\assign}{\coloneqq}
\renewcommand{\leq}{\leqslant}
\renewcommand{\geq}{\geqslant}
\newcommand{\longto}{\longrightarrow}
\newcommand{\stepcomment}[1]{(\textsl{#1})}

\DeclareMathOperator{\Mint}{\mathsf{M}}
\DeclareMathOperator{\Mlo}{\mathsf{M}_{\mathrm{lo}}}
\DeclareMathOperator{\Mhi}{\mathsf{M}_{\mathrm{hi}}}
\DeclareMathOperator{\convcost}{\mathsf{C}}
\DeclareMathOperator{\round}{round}

\DeclareMathSymbol{\wtildesym}{\mathord}{largesymbols}{"65}
\newcommand\lowerwtildesym{%
  \text{\smash{\raisebox{-1.35ex}{%
    $\wtildesym$}}}}
\newcommand\wtilde[1]{%
  \mathchoice
    {\accentset{\displaystyle\lowerwtildesym}{#1}}
    {\accentset{\textstyle\lowerwtildesym}{#1}}
    {\accentset{\scriptstyle\lowerwtildesym}{#1}}
    {\accentset{\scriptscriptstyle\lowerwtildesym}{#1}}
}

\DeclareFontFamily{U}{mathx}{\hyphenchar\font45}
\DeclareFontShape{U}{mathx}{m}{n}{<-> mathx10}{}
\DeclareSymbolFont{mathx}{U}{mathx}{m}{n}
\DeclareMathSymbol{\wbarsym}{\mathord}{mathx}{"73}
\newcommand\lowerwbarsym{%
  \text{\smash{\raisebox{-1.25ex}{%
    $\wbarsym$}}}}
\newcommand\wbar[1]{%
  \mathchoice
    {\accentset{\displaystyle\lowerwbarsym}{#1}}
    {\accentset{\textstyle\lowerwbarsym}{#1}}
    {\accentset{\scriptstyle\lowerwbarsym}{#1}}
    {\accentset{\scriptscriptstyle\lowerwbarsym}{#1}}
}

\begin{document}   

\title{Faster truncated integer multiplication}
\author{David Harvey}
\email{d.harvey@unsw.edu.au}
\address{School of Mathematics and Statistics,
   University of New South Wales, Sydney NSW 2052, Australia}

\begin{abstract}
   We present new algorithms for computing the low $n$ bits or
   the high~$n$ bits of the product of two $n$-bit integers.
   We show that these problems may be solved in asymptotically
   $75\%$ of the time required to compute the full $2n$-bit product,
   assuming that the underlying integer multiplication algorithm relies on
   computing cyclic convolutions of sequences of real numbers.
\end{abstract}

\maketitle

\section{Introduction}

Let $n \geq 1$ and let $u$ and $v$ be integers
in the interval $0 \leq u, v < 2^n$.
We write $\Mint(n)$ for the cost of computing
the \emph{full product} of $u$ and $v$,
which is just the usual $2n$-bit product $uv$.
Unless otherwise specified, by `cost' we mean the number of bit operations,
under a model such as the multitape Turing machine \cite{Pap-complexity}.

In this paper we are interested in two types of \emph{truncated product}.
The \emph{low product} of $u$ and~$v$ is the unique integer $w$
in the interval $0 \leq w < 2^n$ such that $w \equiv uv \pmod{2^n}$,
or in other words, the low $n$ bits of $uv$.
We denote the cost of computing the low product by $\Mlo(n)$.

The \emph{high product} of $u$ and $v$ consists of the high $n$ bits of $uv$,
except that we allow a small error in the lowest bit.
More precisely, the high product is defined to be
any integer $w$ in the range $0 \leq w \leq 2^n$
such that $|uv - 2^n w| < 2^n$.
Thus there are at most two possible values for the high product,
and an algorithm that computes it is permitted to return either one.
We denote the cost of computing the high product by $\Mhi(n)$.

There are many applications of truncated products in computer arithmetic.
The most obvious example is high-precision arithmetic on real numbers:
to compute an $n$-bit approximation to the product of
two real numbers with $n$-bit mantissae,
we may scale by an appropriate power of two
to convert the inputs to $n$-bit integers,
and then compute the high product of those integers.
Further examples include Barrett's \cite{Bar-RSA}
and Montgomery's \cite{Mon-modular} algorithms for modular arithmetic.

It is natural to ask whether a truncated product can be computed
more quickly than a full product.
This is indeed the case for small $n$:
in the classical quadratic-time regime,
we can compute a truncated product in about half the time of a full product,
because essentially only half of the $n^2$ bit-by-bit products
contribute to the desired output.

However, as $n$ grows,
and more sophisticated multiplication algorithms are deployed,
these savings begin to dissipate.
Consider for instance Karatsuba's algorithm,
which has complexity $\Mint(n) = O(n^\alpha)$ for
$\alpha = \log3 / \log 2 \approx 1.58$.
Mulders showed \cite{Mul-short} that Karatsuba's algorithm may be adapted
to obtain bounds for $\Mhi(n)$ and $\Mlo(n)$ around $0.81 \Mint(n)$.
However, it is not known how to reach $0.5 \Mint(n)$ in this regime.

For much larger values of $n$,
the most efficient integer multiplication algorithms known are based on FFTs
(fast Fourier transforms).
Currently, the asymptotically fastest such algorithm has complexity
$\Mint(n) = O(n \log n)$ \cite{HvdH-nlogn},
and it is widely believed that this bound is optimal up to a constant factor.

It has long been thought that the best way to compute a truncated product
using FFT-based algorithms is to simply compute the full product
and then discard the unwanted part of the output.
One might be able to save~$O(n)$ bit operations compared to the full product
by skipping computations that do not contribute
to the desired half of the output,
but no bounds of the type $\Mlo(n) < c \Mint(n)$ or $\Mhi(n) < c \Mint(n)$
have been proved for any constant~$c < 1$.

For some closely related problems,
one can actually \emph{prove} that it is not possible to do better
than computing the full product.
For example, in a suitable algebraic model,
the multiplicative complexity of any algorithm that computes
the low $n$ coefficients of the product of two polynomials
of degree less than $n$ is at least $2n-1$ \cite[Thm.~17.14]{BCS-complexity},
which is the same as the multiplicative complexity of the full product.
By analogy, one might expect the same sort of lower bound
to apply to truncated integer multiplication.

In this paper we show that this belief is mistaken:
we present algorithms that compute high and low products of integers
in asymptotically $75\%$ of the time required for a full product.
The new algorithms require that the underlying integer multiplication
is carried out via a cyclic convolution of sequences of real numbers.
This includes any real convolution algorithm based on FFTs,
and in particular the $O(n \log n)$ multiplication algorithm
of \cite{HvdH-nlogn}.

Unfortunately,
because the new methods rely heavily on the archimedean property of $\RR$,
we do not yet know how to obtain this 25\% reduction in complexity
for arbitrary integer multiplication algorithms.
In particular, we are currently unable to establish analogous results
for integer multiplication algorithms based on FFTs over other rings,
such as finite fields \cite{Pol-ntt}.

Although we focus on time complexity in this paper,
the new techniques also have implications for space complexity.
For example, to multiply two floating-point numbers with $n$-bit mantissae
using standard FFT methods,
the transform of each multiplicand occupies roughly $4n$ bits of storage.
This is true regardless of the type of FFT algorithm used;
it holds for FFTs over finite fields just as well as for real or complex FFTs.
Using the new methods, the storage required drops to roughly $3n$ bits.
This improvement in space complexity may be significant in applications
where storage is the bottleneck,
such as extremely high-precision evaluation of
numerical constants such as $\pi$.
Furthermore, if the computation is I/O-bound,
then this may lead directly to a corresponding improvement in computation time.
We will not explore this issue further in this paper.

The remainder of the paper is structured as follows.
In Section~\ref{sec:setup} we state our main results precisely,
after giving some preliminary definitions.
Section~\ref{sec:low} presents the new algorithm for the low product,
including the proof of correctness and complexity analysis.
Section \ref{sec:high} does the same for the high product.
Section~\ref{sec:code} gives some performance data
for an implementation of the new algorithms.

\smallskip
\emph{Historical note.}
An earlier version of this paper pointed out that the new methods for truncated
multiplication may be used to design an integer multiplication algorithm
having complexity $\Mint(n) = O(n \log n \, K^{\log^* n})$ with $K = 6$.
At the time, this was the best known asymptotic bound for $\Mint(n)$.
This result was subsequently superseded by 
\cite{HvdH-cyclotomic1} ($K = 4 \sqrt 2 \approx 5.66$),
\cite{HvdH-lattice} ($K = 4$),
and then \cite{HvdH-nlogn} ($\Mint(n) = O(n \log n)$).

\section{Setup and statement of results}
\label{sec:setup}

\subsection{Fixed point arithmetic and real convolutions}
\label{sec:arithmetic}

We write $\lg x$ for $\lceil \log_2 x \rceil$.
To simplify analysis of numerical error,
all algorithms are assumed to work with the following
fixed-point representation for real numbers.
(See \cite[\S3]{HvdHL-mul} for a more detailed treatment.)
Let $p \geq 1$ be a precision parameter.
We write $\RR_p$ for the set of real numbers of the form $a/2^p$
where $a$ is an integer in the interval $-2^p \leq a \leq 2^p$.
Thus $\RR_p$ models the unit interval $[-1, 1]$,
and elements of $\RR_p$ are represented using $p + O(1)$ bits of storage.
For $e \in \ZZ$, we write $2^e \RR_p$ for the set of real numbers
of the form $2^e x$ where $x \in \RR_p$.
An element of $2^e \RR_p$ is represented simply by its mantissa in $\RR_p$;
the exponent $e$ is always known from context, and is not explicitly stored.

We will frequently work with quotient rings of the form $\RR[X]/P(X)$
where $P(X)$ is some fixed monic polynomial of positive degree,
such as $X^N - 1$.
If $F \in \RR[X]/P(X)$ and $\deg P = N$,
we write $F_0, \ldots, F_{N-1}$ for the coefficients of~$F$
with respect to the standard monomial basis;
that is, $F = F_0 + \cdots + F_{N-1} X^{N-1} \pmod{P(X)}$.
For such $F$ we define a norm
\begin{equation}
\label{eq:norm-definition}
   \norm{F} \coloneqq \max_{0 \leq i < N} |F_i|.
\end{equation}
We write $2^e \RR_p[X]/P(X)$ for the set of polynomials $F \in \RR[X]/P(X)$
whose coefficients $F_0, \ldots, F_{N-1}$ lie in $2^e \RR_p$;
this is a slight abuse of notation, as $2^e \RR_p$ is not really a ring.
Algorithms always represent such a polynomial by its coefficient vector
$(F_0, \ldots, F_{N-1}) \in (2^e \RR_p)^N$.

We assume that we have available a subroutine \textsc{Convolution}
with the following properties.
It takes as input two parameters $N \geq 2$ and $p \geq 1$, and polynomials
\[ F, G \in 2^e \RR_p[X]/(X^N - 1). \]
Let $H \coloneqq FG \in \RR[X]/(X^N - 1)$; more explicitly,
\[ H_k \coloneqq \sum_{i+j \equiv k \bmod N} F_i G_j \quad
   \in \mathopen{\big[}-2^{2e} N, 2^{2e} N\big], \qquad 0 \leq k < N. \]
Then \textsc{Convolution} is required to output a polynomial
\[ \wtilde H \in 2^{2e+ \lg N}\RR_p[X]/(X^N - 1) \]
such that
\[ \norm{\wtilde H - H} < 2^{2e + \lg N - p}. \]
In other words, \textsc{Convolution} computes a $p$-bit approximation
to the cyclic convolution of two real input sequences of length $N$.

We write $\convcost(N, p)$ for the bit complexity of \textsc{Convolution}.
We treat this routine as a black box;
its precise implementation is not important for our purposes.
A typical implementation would execute a
real-to-complex FFT for each input sequence,
multiply the Fourier coefficients pointwise,
and then compute an inverse complex-to-real transform to recover the result.
Internally, it should work to precision slightly higher than $p$
to control rounding errors during intermediate computations.
(For an explicit error bound, see for example \cite[Theorem 3.6]{BZ-mca}.)
The routine may completely ignore the exponent parameter $e$.

\subsection{The full product}

For completeness,
we recall the well-known algorithm that uses \textsc{Convolution} to compute
the full product of two $n$-bit integers (Algorithm~\ref{algo:full} below).
It depends on two parameters:
a chunk size~$b$, and a transform length~$N$, where $Nb \geq n$.
The idea is to cut the integers into $N$ chunks of $b$ bits,
thereby reducing the integer multiplication problem to the problem of
multiplying two polynomials in $\ZZ[X]$ modulo $X^{2N} - 1$.

We will not discuss in this paper the question of
optimising the choice of $b$ and~$N$.
The optimal choice of~$N$ will involve some balance between
making~$N$ as close to~$n/b$ as possible,
but also ensuring that~$N$ is sufficiently smooth
(has only small prime factors)
so that FFTs of length~$N$ are as efficient as possible.
(An alternative approach is to use ``truncated FFTs''~\cite{vdH-TFT-apps},
which eliminates the need to choose a smooth transform length.
However, this makes no difference asymptotically.
Despite the overlapping terminology,
it is not clear whether the new truncated multiplication algorithms
can be adapted to the case of truncated FFTs.
This is an interesting question for future research.)

\begin{algorithm}
   \SetAlgoLined
   \DontPrintSemicolon
   \KwIn{Parameters $n \geq 1$, $b \geq 1$, $N \geq 2$ with $N b \geq n$,
      \newline integers $0 \leq u, v < 2^n$.}
   \KwOut{$uv$ (the full product of $u$ and $v$).}
   $p \assign 2b + \lg N + 2$.\;
   \stepcomment{Split inputs}
   Compute $u_0, \ldots, u_{N-1}$ and $v_0, \ldots, v_{N-1}$
   with $0 \leq u_i, v_i < 2^b$ such that
   \[ u = \sum_{i=0}^{N-1} u_i 2^{ib},
      \qquad v = \sum_{i=0}^{N-1} v_i 2^{ib}, \]
   and let
   \[ \wbar U(X) \assign \sum_{i=0}^{N-1} u_i X^i,
      \qquad \wbar V(X) \assign \sum_{i=0}^{N-1} v_i X^i,
      \qquad \wbar U, \wbar V \in 2^b\RR_p[X]/(X^{2N} - 1). \]
   \; \vspace{-10pt}
   \stepcomment{Perform convolution}
   Use \textsc{Convolution} to compute
   \[ \wbar W \in 2^{2b+\lg 2N}\RR_p[X]/(X^{2N} - 1) \]
   such that
   \[\norm{\wbar W - \wbar U \wbar V} < 2^{2b+\lg 2N-p}. \]
   \; \vspace{-10pt}
   \stepcomment{Overlap-add}
   \KwRet{$\sum_{i=0}^{2N-2} \round(\wbar{W_i}) \, 2^{ib}$}. \;
\caption{Full product}
\label{algo:full}
\end{algorithm}

\begin{thm}[Full product]
\label{thm:full}
   Let $n \geq 1$, and let $u$ and $v$ be $n$-bit integers.
   Let $b \geq 1$ and $N \geq 2$ be integers such that $N b \geq n$.
   Then Algorithm~\ref{algo:full} correctly computes
   the full product of $u$ and $v$.
   Assuming that $\lg N = O(b)$, its complexity is
   \[ \Mint(n) = \convcost(2N, 2b + \lg N + 2) + O(Nb). \]
\end{thm}
\begin{proof}
   The condition $Nb \geq n$ ensures that the decompositions of $u$ and $v$
   into $u_0, \ldots, u_{N-1}$ and $v_0, \ldots, v_{N-1}$ in line 2 are legal.
   Let
   \[ U(X) \coloneqq \sum_{i=0}^{N-1} u_i X^i \in \ZZ[X],
      \qquad V(X) \coloneqq \sum_{i=0}^{N-1} v_i X^i \in \ZZ[X], \]
   and 
   \[ W(X) \coloneqq U(X) V(X) = \sum_{i=0}^{2N-2} w_i X^i \in \ZZ[X]. \]
   Note that $\wbar U$ and $\wbar V$ are the images of $U$ and $V$ in
   $2^b \RR_p[X]/(X^{2N} - 1)$,
   and by construction $u = U(2^b)$ and $v = V(2^b)$.
   Since $W(X)$ has degree at most $2N-2$,
   it is determined by its remainder modulo $X^{2N} - 1$.
   Line~3 computes an approximation~$\wbar W$ to this remainder with
   \[ |\wbar{W_i} - w_i| < 2^{2b+\lg2N-p} = \tfrac12 \]
   for each~$i$.
   The function $\round(\mspace{2mu}\cdot\mspace{2mu})$ in line 4
   rounds its argument to the nearest integer,
   with ties broken in either direction as convenient.
   Since $w_i \in \ZZ$, we deduce that $\round(\wbar{W_i}) = w_i$ for each $i$;
   hence line 4 returns $W(2^b) = U(2^b) V(2^b) = uv$.

   The main term in the complexity bound arises from the
   \textsc{Convolution} call in line 3.
   The secondary term consists of the splitting step in line 2,
   which costs $O(b)$ bit operations per coefficient,
   and the overlap-add procedure in line 4,
   which requires $O(b + \lg N) = O(b)$ bit operations per coefficient.
\end{proof}

\subsection{Statement of results}
\label{sec:results}

The main results of this paper are the following analogues
of Theorem \ref{thm:full} for the low product and high product.
These are proved in Section \ref{sec:low} and
Section \ref{sec:high} respectively.
\begin{thm}[Low product]
\label{thm:low}
   Let $n \geq 1$, and let $u$ and $v$ be $n$-bit integers.
   Let $b \geq 4$ and $N \geq 3$ be integers such that $N b \geq n$.
   Then Algorithm~\ref{algo:low} (see \S\ref{sec:low})
   correctly computes the low product of $u$ and $v$.
   Assuming that $\lg N = O(b)$, its complexity is
   \begin{equation}
   \label{eq:low-bound}
      \Mlo(n) = \convcost(N, 3b + \lg N + 6) + O(N \Mint(b)).
   \end{equation}
\end{thm}
\begin{thm}[High product]
\label{thm:high}
   Let $n \geq 1$, and let $u$ and $v$ be $n$-bit integers.
   Let $b \geq 4$ and $N \geq 3$ be integers such that
   $(N + 1)b \geq n + \lg N + 2$.
   Then Algorithm~\ref{algo:high} (see \S\ref{sec:high})
   correctly computes the high product of $u$ and $v$.
   Assuming that $\lg N = O(b)$, its complexity is
   \begin{equation}
   \label{eq:high-bound}
      \Mhi(n) = \convcost(N, 3b + \lg N + 9) + O(N \Mint(b)).
   \end{equation}
\end{thm}

Comparing these complexity bounds to Theorem \ref{thm:full} (the full product),
we observe that the convolution length has dropped from $2N$ to $N$,
but the working precision has increased from roughly $2b$ to roughly~$3b$.
To understand the implications for the overall complexity,
we need to make further assumptions on the behaviour of $\convcost(N, p)$.
We consider two scenarios.

\medskip
\textsl{Scenario \#1: asymptotic behaviour as $n \to \infty$.}
Assume that the transform length $N$ is restricted to suitably smooth values,
such as the \emph{ultrasmooth numbers} defined by
Bernstein \cite{Ber-redundancy},
so that asymptotically 100\% of the FFT work is performed
using radix-2 transforms.
Assume also that the working precision $p$ is exponentially smaller than $N$,
but somewhat larger than $\log N$,
say $p = \Theta(\log^2 N)$.
Under these assumptions it is reasonable to expect that the complexity of
the underlying real convolution is quasi-linear with respect to the
the total bit size $Np$.
In particular, for any absolute constants $k \geq 1$ and $c > 0$,
where $k$ is an integer, we should have
\[ \convcost(kN, cp) = (kc + o(1)) \convcost(N, p). \]
This is the case for all FFT-based convolution algorithms known to the author.

Now, given some large $n$,
assume that we choose $N$ and $b$ such that
$Nb = (1 + o(1)) n$ and $b = \Theta(\log^2 n)$,
as is done for example in \cite[\S6]{HvdHL-mul}.
Then
\[ 2b + \lg N + 2 = (2 + o(1)) b, \]
so according to Theorem \ref{thm:full} we have
\[ \Mint(n) = \convcost(2N, 2b + \lg N + 2) + O(Nb)
            = (4 + o(1)) \convcost(N, b). \]
On the other hand, for the low product, Theorem \ref{thm:low} yields
\[ \Mlo(n) = \convcost(N, 3b + \lg N + 6) + O(N \Mint(b))
            = (3 + o(1)) \convcost(N, b). \]
We conclude that
\[ \frac{\Mlo(n)}{\Mint(n)} = \frac{3}{4} + o(1), \]
justifying our assertion that asymptotically the new low product algorithm
saves 25\% of the total work compared to the full product.
Similar remarks apply to the high product.

\medskip
\textsl{Scenario \#2: fixed word size.}
Now let us consider the situation faced by a programmer working on
a modern microprocessor with hardware support for a fixed word size,
such as the 53-bit double-precision floating point type
provided by the IEEE 754 standard.
In this setting,
the \textsc{Convolution} subroutine takes as input two vectors of coefficients
represented by this data type,
and computes their cyclic convolution using some sort of FFT,
taking direct advantage of the hardware arithmetic.
We assume that $b$ is chosen as large as possible so that
the FFTs can be performed in this way;
for example, under IEEE 754 we would require that
$3b + \lg N + \beta_N \leq 53$ for the low product,
where $\beta_N$ is an allowance for numerical error.
Obviously in this scenario it does not make sense to allow $n \to \infty$,
and it also does not quite make sense to measure complexity by the
number of ``bit operations''.
Instead, $n$ should be restricted to lie in some finite
(possibly quite large) range,
and a more natural measure of complexity is the number of word operations
(ignoring issues such as locality and parallelism).

We claim that it is still reasonable to expect
a reduction in complexity close to 25\%.
To see this, consider a full product computation for a given $n$,
with splitting parameters $N$ and $b$.
Let $N'$ and $b'$ be the splitting parameters for the
corresponding truncated product,
for the same value of $n$.
We should choose $b'$ around $2b/3$ to ensure that we still take
maximum advantage of the available floating-point type.
Then we should choose $N'$ around $3N/2$ to compensate for the smaller chunks.
Now observe that (for large $n$) the bulk of the work for the full product
consists of FFTs of length $2N$,
but for the truncated products the FFT length is reduced to around $3N/2$.
Since the FFTs run in quasilinear time (i.e., word operations),
we expect to see roughly 25\% savings.

In practice the expected 25\% speedup will be tempered somewhat by the
additional linear-time work inherent in the truncated product algorithms,
such as the evaluation of $\alpha^*$ and $\beta^*$ in Algorithm \ref{algo:low}.
The situation is also complicated by the fact that we are constrained
to choose smooth transform lengths.
Section \ref{sec:code} gives some timings,
showing the speedup actually achieved by an implementation.

\section{The low product}
\label{sec:low}

The aim of this section is to prove Theorem \ref{thm:low}.
Throughout the section we fix integers 
\begin{equation}
\label{eq:b-N-bound}
   b \geq 4, \qquad N \geq 3
\end{equation}
as in the statement of the theorem.

\subsection{The cancellation trick}

The key to the new low product algorithm is the following simple observation.
\begin{prop}
\label{prop:low-cancel}
   Let $W(X) \in \ZZ[X]$ with $\deg W \leq 2N - 2$, say
   \[ W(X) = \sum_{i=0}^{2N-2} w_i X^i. \]
   Let $L(X) \in \RR[X]$ be the remainder on dividing $W(X)$ by
   \[ A(X) \coloneqq X^N + 2^{-b} X - 1 \in \RR[X], \]
   with $\deg L < N$.
   Then $2^b L(X) \in \ZZ[X]$ and
   \[ L(2^b) = \sum_{i=0}^{N-1} w_i 2^{ib}. \]
\end{prop}
\begin{proof}
   Write
   \[ W(X) = \sum_{i=0}^{2N-2} w_i X^i
         = \sum_{i=0}^{N-1} w_i X^i + \sum_{i=0}^{N-2} w_{N+i} X^{N+i}. \]
   Since $X^N \equiv 1 - 2^{-b} X \pmod{A(X)}$, we have
   \[ W(X) \equiv \sum_{i=0}^{N-1} w_i X^i +
      \sum_{i=0}^{N-2} w_{N+i} (1 - 2^{-b} X) X^i \pmod{A(X)}. \]
   The polynomial on the right hand side has degree at most $N-1$,
   so we deduce that
   \[ L(X) = \sum_{i=0}^{N-1} w_i X^i +
      \sum_{i=0}^{N-2} w_{N+i} (1 - 2^{-b} X) X^i. \]
   This shows that $2^b L(X) \in \ZZ[X]$,
   and the result follows on substituting $X = 2^b$.
\end{proof}

\begin{figure}
\begin{tikzpicture}
   \newcommand{\ystep}{0.5}
   \newcommand{\xstep}{0.7}
   \newcommand{\width}{1.6}
   \draw[right] (0.2*\xstep,0.5*\ystep) node {$1$};
   \draw (        0,0*\ystep) rectangle ++(-\width,\ystep) node[pos=0.5] {$w_0$};
   \draw[right] (0.2*\xstep,2.5*\ystep) node {$X$};
   \draw (  -\xstep,2*\ystep) rectangle ++(-\width,\ystep) node[pos=0.5] {$w_1$};
   \draw[right] (0.2*\xstep,4.5*\ystep) node {$X^2$};
   \draw (-2*\xstep,4*\ystep) rectangle ++(-\width,\ystep) node[pos=0.5] {$w_2$};
   \draw[right] (0.2*\xstep,6.5*\ystep) node {$X^3$};
   \draw (-3*\xstep,6*\ystep) rectangle ++(-\width,\ystep) node[pos=0.5] {$w_3$};
   \draw[right] (0.2*\xstep,8.5*\ystep) node {$X^4$};
   \draw (-4*\xstep,8*\ystep) rectangle ++(-\width,\ystep) node[pos=0.5] {$w_4$};
   \draw[dotted] (        0,9.5*\ystep) -- (        0,-0.5*\ystep) node[below] {$0$};
   \draw[dotted] (  -\xstep,9.5*\ystep) -- (  -\xstep,-0.5*\ystep) node[below] {$b$};
   \draw[dotted] (-2*\xstep,9.5*\ystep) -- (-2*\xstep,-0.5*\ystep) node[below] {$2b$};
   \draw[dotted] (-3*\xstep,9.5*\ystep) -- (-3*\xstep,-0.5*\ystep) node[below] {$3b$};
   \draw[dotted] (-4*\xstep,9.5*\ystep) -- (-4*\xstep,-0.5*\ystep) node[below] {$4b$};
   \draw (-3*\xstep,-1.3) node{(A)};
\end{tikzpicture}
\qquad\qquad
\begin{tikzpicture}
   \newcommand{\ystep}{0.5}
   \newcommand{\xstep}{0.7}
   \newcommand{\width}{1.6}
   \draw[right]        (0.2*\xstep,1.0*\ystep) node {$1$};
   \draw               (         0,0.0*\ystep) rectangle ++(-\width,\ystep) node[pos=0.5] {$w_0$};
   \draw[fill=black!10](         0,1.0*\ystep) rectangle ++(-\width,\ystep) node[pos=0.5] {$w_3$};
   \draw[right]        (0.2*\xstep,4.5*\ystep) node {$X$};
   \draw[fill=black!10](         0,3.0*\ystep) rectangle ++(-\width,\ystep) node[pos=0.5] {$-w_3$};
   \draw               (   -\xstep,4.0*\ystep) rectangle ++(-\width,\ystep) node[pos=0.5] {$w_1$};
   \draw[fill=black!10](   -\xstep,5.0*\ystep) rectangle ++(-\width,\ystep) node[pos=0.5] {$w_4$};
   \draw[right]        (0.2*\xstep,8.0*\ystep) node {$X^2$};
   \draw[fill=black!10](   -\xstep,7.0*\ystep) rectangle ++(-\width,\ystep) node[pos=0.5] {$-w_4$};
   \draw               ( -2*\xstep,8.0*\ystep) rectangle ++(-\width,\ystep) node[pos=0.5] {$w_2$};
   \draw[dotted]       (         0,9.5*\ystep) -- (        0,-0.5*\ystep) node[below] {$0$};
   \draw[dotted]       (   -\xstep,9.5*\ystep) -- (  -\xstep,-0.5*\ystep) node[below] {$b$};
   \draw[dotted]       ( -2*\xstep,9.5*\ystep) -- (-2*\xstep,-0.5*\ystep) node[below] {$2b$};
   \draw (-2*\xstep,-1.3) node{(B)};
\end{tikzpicture}
\caption{
   Illustration of the cancellation trick for $N = 3$.
   Diagram (A) shows the placement of the contributions of
   $w_0, \ldots, w_4$ in $W(2^b)$,
   and similarly (B) shows their contributions to $L(2^b)$,
   where notation is as in Proposition \ref{prop:low-cancel}.
   Observe that $w_3$ and $w_4$ both appear twice in (B)
   (the shaded cells),
   due to the reduction modulo $A(X)$.
   Thanks to the choice of signs,
   their overall contribution to $L(2^b)$ is zero.}
\label{fig:low-cancel}
\end{figure}

Later we will apply Proposition \ref{prop:low-cancel} to
a polynomial $W(X) = U(X) V(X)$ analogous to the $W(X)$ encountered earlier
in the proof of Theorem \ref{thm:full}.
The proposition shows that after reducing $W(X)$ modulo $A(X)$ and
making the substitution $X = 2^b$,
the $2^{-b} X$ term in $A(X)$ causes the unwanted high-order coefficients
of $W(X)$ to disappear;
see Figure \ref{fig:low-cancel}.
An alternative point of view is that polynomial multiplication modulo $A(X)$
corresponds roughly to integer multiplication modulo
\[ A(2^b) = 2^{Nb} + 2^{-b} 2^b - 1 = 2^{Nb}. \]

To make use of Proposition \ref{prop:low-cancel} to compute a low product,
we must compute $L(X)$ \emph{exactly}.
Note that the coefficients of $L(X)$ lie in $2^{-b}\ZZ$ rather than $\ZZ$.
Consequently, to compute $L(X)$,
we must increase the working precision by $b$ bits compared to the precision
used in the full product algorithm.
This is why the precision parameter in Theorem \ref{thm:low}
(and Theorem \ref{thm:high}) is $3b + \lg N + O(1)$
rather than $2b + \lg N + O(1)$.

\subsection{The roots of $A(X)$}
\label{sec:low-roots}

In this section we study the complex roots of the special polynomial $A(X)$
introduced in Proposition \ref{prop:low-cancel}.
For $r > 0$, let $D_r$ denote the open disc $\{z \in \CC : |z| < r\}$.
\begin{lem}
\label{lem:low-roots}
   The roots of $A(X)$ lie in $D_2$, and they are all simple.
\end{lem}
\begin{proof}
   If $z \in \CC$ is a root of $A(X)$ and $|z| \geq 2$,
   then \eqref{eq:b-N-bound} implies that
   \[ 4|z| \leq |z|^3 \leq |z|^N = |z^N| = |1 - 2^{-b}z|
      \leq 1 + \frac{|z|}{16} \leq 2|z|, \]
   which is impossible.

   Any multiple root $z$ of $A(X)$ would have to satisfy
   \[ A(z) = z^N + 2^{-b} z - 1 = 0, \qquad A'(z) = N z^{N-1} + 2^{-b} = 0, \]
   and hence
   \[ 0 = N A(z) - z A'(z) = 2^{-b} (N-1) z - N. \]
   This implies that $z > 0$, contradicting $A'(z) = 0$.
\end{proof}

Now consider the function
\[ \beta(z) \coloneqq z (1 - 2^{-b} z)^{-1/N}, \qquad z \in D_{2^b}, \]
where $u \mapsto u^{-1/N}$ means the branch that maps $1$ to $1$.
\begin{lem}
\label{lem:low-maps}
   The function $\beta(z)$ maps roots of $A(X)$ to roots of $X^N - 1$.
\end{lem}
\begin{proof}
   If $z$ is a root of $A(X) = X^N + 2^{-b} X - 1$, then
   \[ \beta(z)^N = \frac{z^N}{1-2^{-b} z} = \frac{z^N}{z^N} = 1. \qedhere \]
\end{proof}
In fact, $\beta(z)$ always sends a root of $A(X)$
to the root of $X^N - 1$ nearest to it,
but we will not prove this.
Figure \ref{fig:low-roots} illustrates the situation for $N = 12$ and $b = 1$,
showing that the roots of $A(X)$ are very close to those of $X^N - 1$.
(For $b = 2$ the roots are already too close together to distinguish
at this scale.)
\begin{figure}
\begin{tikzpicture}[scale=2]
   \draw (-1.2, 0) -- (1.2, 0);
   \draw (0, -1.2) -- (0, 1.2);
   \draw (-1.000, 0.000) circle(0.02);
   \draw (1.000, 0.000) circle(0.02);
   \draw (-0.866, -0.500) circle(0.02);
   \draw (-0.866, 0.500) circle(0.02);
   \draw (-0.500, -0.866) circle(0.02);
   \draw (-0.500, 0.866) circle(0.02);
   \draw (0.000, -1.000) circle(0.02);
   \draw (0.000, 1.000) circle(0.02);
   \draw (0.500, -0.866) circle(0.02);
   \draw (0.500, 0.866) circle(0.02);
   \draw (0.866, -0.500) circle(0.02);
   \draw (0.866, 0.500) circle(0.02);
   \filldraw[black] (-1.035, 0.000) circle(0.02);
   \filldraw[black] (0.948, 0.000) circle(0.02);
   \filldraw[black] (-0.886, -0.530) circle(0.02);
   \filldraw[black] (-0.886, 0.530) circle(0.02);
   \filldraw[black] (-0.486, -0.901) circle(0.02);
   \filldraw[black] (-0.486, 0.901) circle(0.02);
   \filldraw[black] (0.040, -1.007) circle(0.02);
   \filldraw[black] (0.040, 1.007) circle(0.02);
   \filldraw[black] (0.529, -0.832) circle(0.02);
   \filldraw[black] (0.529, 0.832) circle(0.02);
   \filldraw[black] (0.847, -0.454) circle(0.02);
   \filldraw[black] (0.847, 0.454) circle(0.02);
\end{tikzpicture}
\caption{Roots of $A(X) = X^{12} + X/2 - 1$ (filled circles)
   and $X^{12} - 1$ (empty circles).}
\label{fig:low-roots}
\end{figure}

For any $k \in \ZZ$,
the binomial theorem implies that $\beta(z)^k$ is represented
on $D_{2^b}$ by the series
\[ \beta(z)^k = z^k \sum_{r=0}^\infty \beta_{k,r} z^r
              = z^k + \beta_{k,1} z^{k+1} + \beta_{k,2} z^{k+2} + \cdots \]
where
\[ \beta_{k,r} \coloneqq \binom{-k/N}{r} (-2^{-b})^r, \qquad r \geq 0. \]
In particular, the first few terms of $\beta(z)$ are given by
\[ \beta(z) = z + \frac{1}N 2^{-b} z^2
         + \frac{(N+1)}{2N^2}2^{-2b} z^3
         + \frac{(N+1)(2N+1)}{6N^3} 2^{-3b} z^4 + \cdots. \]

We will need to construct an explicit functional inverse for $\beta(z)$,
in order to map the roots of $X^N - 1$
back to the corresponding roots of $A(X)$.
Let $\alpha(z) \in z \, \RR[[z]]$
be the formal power series inverse of $\beta(z)$,
i.e., so that
\[ \beta(\alpha(z)) = z = \alpha(\beta(z)). \]
The coefficients of $\alpha(z)$, and of its powers, are given as follows.
\begin{lem}
\label{lem:alpha-coeffs}
   For any $k \geq 0$ we have (formally)
   \[ \alpha(z)^k = z^k \sum_{r=0}^\infty \alpha_{k,r} z^r
            = z^k +  \alpha_{k,1} z^{k+1} + \alpha_{k,2} z^{k+2} + \cdots, \]
   where $\alpha_{k,0} \coloneqq 1$ and
   \[ \alpha_{k,r} \coloneqq \frac{k}{k+r} \binom{(k+r)/N}{r} (-2^{-b})^r,
         \qquad r \geq 1. \]
\end{lem}
In particular, the first few terms of $\alpha(z)$ are
\[ \alpha(z) = z - \frac{1}N 2^{-b} z^2 - \frac{(N-3)}{2 N^2} 2^{-2b} z^3
         - \frac{(N-4)(2N-4)}{6 N^3} 2^{-3b} z^4 - \cdots. \]
\begin{proof}
By the Lagrange inversion formula \cite[Thm.~5.4.2]{Sta-EC2},
for any $n \geq k$ we have
\[ n \alpha_{k,n-k} = k \beta_{-n,n-k}. \]
Taking $n \coloneqq k + r$, for any $r \geq 1$, yields
\[ \alpha_{k,r} = \frac{k}{k+r} \beta_{-k-r,r}
                = \frac{k}{k+r} \binom{(k+r)/N}{r} (-2^{-b})^r. \qedhere \]
\end{proof}

\begin{rem}
   It is also possible to write down
   an explicit formula for $\alpha_{k,r}$ when $k < 0$,
   but the above argument fails because $k+r$ is zero when $r = -k$.
   To handle the $k < 0$ case one needs a slightly stronger form of the
   Lagrange inversion formula;
   see for example \cite[Thm.~2.1.1]{Ges-lagrange}.
   In this paper we only need the case $k \geq 0$.
\end{rem}

The next result gives some simple bounds for the coefficients
$\alpha_{k,r}$ and $\beta_{k,r}$.
\begin{lem}
\label{lem:alpha-beta-bound}
   For all $r \geq 0$ and $0 \leq k < N$ we have
   \[ |\beta_{k,r}| \leq 2^{-rb}, \qquad |\alpha_{k,r}| \leq 2^{-rb}. \]
\end{lem}
\begin{proof}
   The bounds are trivial for $r = 0$, so assume that $r \geq 1$.
   For $\beta_{k,r}$ we have
   \[ \frac{|\beta_{k,r}|}{2^{-rb}}
      = \frac{1}{r!} \prod_{j=0}^{r-1} \left| \frac{-k}{N} - j \right|
      = \frac{1}{r!} \prod_{j=0}^{r-1} \left(j + \frac{k}{N}\right)
      \leq \frac{1}{r!} \prod_{j=0}^{r-1} (j+1) = 1. \]
   For $\alpha_{k,r}$, observe that
   \[ \frac{|\alpha_{k,r}|}{2^{-rb}}
      = \frac{k}{k+r} \cdot \frac{1}{r!}
            \prod_{j=0}^{r-1} \left| \frac{k+r}{N} - j \right|
      \leq \frac{1}{r!} \prod_{j=0}^{r-1} | j - \eta | \]
   where $\eta \coloneqq (k+r)/N$.
   We have $\eta > 0$ and
   \[ \eta = \frac{k+r}{N} \leq \frac{N-1+r}{N} = \frac{r-1}{N} + 1
   \leq \frac{r-1}{3} + 1 \leq r \]
   since $r \geq 1$ and $N \geq 3$ (see \eqref{eq:b-N-bound});
   hence $1 \leq \lceil \eta \rceil \leq r$.
   Thus
   \begin{align*}
      \prod_{j=0}^{r-1} |j - \eta|
         & = \prod_{0 \leq j < \lceil \eta \rceil} (\eta - j)
             \prod_{\lceil \eta \rceil \leq j \leq r-1} (j - \eta) \\
         & \leq \prod_{0 \leq j < \lceil \eta \rceil} (\lceil \eta \rceil - j)
            \prod_{\lceil \eta \rceil \leq j \leq r-1} (j + 1 - \lceil \eta \rceil) \\
         & = \lceil \eta \rceil! (r - \lceil \eta \rceil)!
         = r! \Big\slash \binom{r}{\lceil \eta \rceil} \leq r!. \qedhere
   \end{align*}
\end{proof}

\begin{cor}
\label{cor:low-inverses}
   The series for $\alpha(z)$ and $\beta(z)$ converge on $D_{2^b}$, and
   \begin{equation}
   \label{eq:low-inverses}
      \alpha(\beta(z)) = z = \beta(\alpha(z)), \qquad z \in D_2.
   \end{equation}
\end{cor}
\begin{proof}
   We already know that $\beta(z)$ converges on $D_{2^b}$,
   and the convergence of $\alpha(z)$ on $D_{2^b}$
   follows from Lemma \ref{lem:alpha-beta-bound}.
   If $|z| < 2$, then
   \[ |\alpha(z)| \leq \sum_{r=0}^\infty |\alpha_{1,r}| |z|^{r+1}
      < \sum_{r=0}^\infty 2^{-rb} 2^{r+1}
      = \frac{2}{1 - 2^{-b+1}} < 3, \]
   where the last inequality follows from \eqref{eq:b-N-bound}.
   This shows that $\alpha(z)$ maps $D_2$ into $D_3 \subseteq D_{2^b}$.
   A similar argument shows that $\beta(z)$ maps $D_2$
   into $D_3 \subseteq D_{2^b}$.
   Since both $\alpha(z)$ and $\beta(z)$ map $D_2$ into
   the disc of convergence of the other,
   and since they are inverses formally,
   they must be inverse functions in the sense of \eqref{eq:low-inverses}.
\end{proof}

\begin{cor}
\label{cor:low-bijection}
   The functions $\alpha(z)$ and $\beta(z)$ induce mutually inverse bijections
   between the roots of $X^N - 1$ and the roots of $A(X)$.
\end{cor}
\begin{proof}
   By Lemma \ref{lem:low-roots},
   the polynomial $A(X)$ has $N$ distinct roots $z_1, \ldots, z_N$ in $D_2$.
   Lemma \ref{lem:low-maps} shows that $\beta(z)$ maps $z_1, \ldots, z_N$
   to roots of $X^N - 1$,
   and the images must be distinct because Corollary \ref{cor:low-inverses}
   implies that $\beta(z)$ is injective on $D_2$.
   Since $X^N - 1$ has exactly $N$ roots,
   every root of $X^N - 1$ must be the image of some $z_i$,
   and then $\alpha(z)$ must map this root back to $z_i$.
\end{proof}

\subsection{Ring isomorphisms}
\label{sec:low-iso}

The aim of this section is construct a pair of
mutually inverse ring isomorphisms
\begin{equation}
\label{eq:alpha-beta-star}
   \begin{split}
      \alpha^* & \colon \RR[X]/A(X) \longto \RR[X]/(X^N - 1), \\
      \beta^*  & \colon \RR[X]/(X^N - 1) \longto \RR[X]/A(X).
   \end{split}
\end{equation}
In the main low product algorithm,
the role of these maps will be to convert the problem of
multiplying two polynomials modulo $A(X)$ into an ordinary cyclic convolution.

The idea of the construction is that for $F \in \RR[X]/A(X)$,
we want to define $(\alpha^* F)(X)$ to be the composition $F(\alpha(X))$,
regarded as a polynomial modulo ${X^N - 1}$,
and similarly for $\beta^*$.
However, some care is required in interpreting the expression $F(\alpha(X))$,
as $\alpha(z)$ is not a polynomial, but rather a power series.
To make this definition precise, we proceed as follows.

For each $r \geq 0$ define linear maps
\begin{align*}
   \alpha^*_r & \colon \RR[X]/A(X) \longto \RR[X]/(X^N - 1), \\
   \beta^*_r &  \colon \RR[X]/(X^N - 1) \longto \RR[X]/A(X)
\end{align*}
by the formulas
\begin{align*}
   \alpha^*_r \bigg(\sum_{k=0}^{N-1} F_k X^k \bmod{A(X)}\bigg)
      & \coloneqq \sum_{k=0}^{N-1} \alpha_{k,r} F_k X^{k+r} \bmod{X^N - 1}, \\
   \beta^*_r \bigg(\sum_{k=0}^{N-1} F_k X^k \bmod{X^N - 1}\bigg)
      & \coloneqq \sum_{k=0}^{N-1} \beta_{k,r} F_k X^{k+r} \bmod{A(X)}.
\end{align*}
These maps satisfy the following norm bounds.
(Recall that the norm on polynomials is defined as in
\eqref{eq:norm-definition}.)
\begin{lem}
\label{lem:alpha-r-bound}
   For any $r \geq 0$ and $F \in \RR[X]/A(X)$,
   \[ \norm{\alpha^*_r F} \leq 2^{-rb} \norm{F}. \]
\end{lem}
\begin{proof}
   By definition, $\alpha^*_r F = X^r G$ where
   \[ G \coloneqq \sum_{k=0}^{N-1} \alpha_{k,r} F_k X^k \in \RR[X]/(X^N - 1). \]
   For any $H \in \RR[X]/(X^N - 1)$ we have $\norm{XH} = \norm{H}$,
   because multiplication by $X$ simply permutes the coefficients cyclically.
   Applying this observation to $G$ repeatedly, 
   and recalling Lemma~\ref{lem:alpha-beta-bound}, we find that
   \[ \norm{\alpha^*_r F} = \norm{X^r G} = \norm{G} \leq 2^{-rb} \norm{F}.
      \qedhere \]
\end{proof}

\begin{lem}
\label{lem:beta-r-bound}
   For any $r \geq 0$ and $F \in \RR[X]/(X^N - 1)$,
   \[ \norm{\beta^*_r F} \leq 2^{-r(b-1)} \norm{F}. \]
\end{lem}
\begin{proof}
   The argument is similar to Lemma \ref{lem:alpha-r-bound},
   the main difference being that multiplication by $X$ modulo $A(X)$
   is slightly more complicated than a cyclic permutation.
   Let $H = \sum_{k=0}^{N-1} H_k X^k \in \RR[X]/A(X)$.
   Since $X^N \equiv 1 - 2^{-b} X \pmod{A(X)}$, we have
   \[ XH = H_{N-1} + (H_0 - 2^{-b} H_{N-1}) X
         + H_1 X^2 + \cdots + H_{N-2} X^{N-1}, \]
   so $\norm{XH} \leq (1 + 2^{-b}) \norm{H} \leq 2 \norm{H}$.
   
   Now, since $\beta^*_r F = X^r G$ where
   $G \coloneqq \sum_{k=0}^{N-1} \beta_{k,r} F_k X^k \in \RR[X]/A(X)$,
   using Lemma \ref{lem:alpha-beta-bound} we obtain
   \[ \norm{\beta^*_r F} = \norm{X^r G} \leq 2^r \norm{G}
      \leq 2^r 2^{-rb} \norm{F}. \qedhere \]
\end{proof}

We may now define the maps $\alpha^*$ and $\beta^*$
in \eqref{eq:alpha-beta-star} by setting
\[ \alpha^* F \coloneqq \sum_{r=0}^\infty \alpha^*_r F, \qquad
   \beta^* F \coloneqq \sum_{r=0}^\infty \beta^*_r F. \]
Lemma \ref{lem:alpha-r-bound} and Lemma \ref{lem:beta-r-bound} guarantee
that these series converge coefficientwise,
so $\alpha^*$ and~$\beta^*$ are well-defined,
and they are clearly linear maps.
Moreover, we immediately obtain the following estimates
concerning the partial sums of the series.
\begin{lem}
\label{lem:alpha-series}
   For any $F \in \RR[X]/A(X)$ and any integer $\lambda \geq 0$ we have
   \[ \bignorm{\alpha^*F - \sum_{r=0}^{\lambda-1} \alpha^*_r F}
         \leq \frac{16}{15} \cdot 2^{-\lambda b} \norm{F}, \qquad
      \bignorm{\sum_{r=0}^{\lambda-1} \alpha^*_r F}
         \leq \frac{16}{15} \norm{F}. \]
\end{lem}
\begin{proof}
For the first claim, observe that
\begin{multline*}
   \bignorm{\alpha^*F - \sum_{r=0}^{\lambda-1} \alpha^*_r F}
      = \bignorm{\sum_{r=\lambda}^\infty \alpha^*_r F}
      \leq \sum_{r=\lambda}^\infty \norm{\alpha^*_r F}
      \leq \sum_{r=\lambda}^\infty 2^{-rb} \norm{F} \\
   = \frac{2^{-\lambda b}}{1 - 2^{-b}} \norm{F}
   \leq \frac{16}{15} \cdot 2^{-\lambda b} \norm{F},
\end{multline*}
by Lemma \ref{lem:alpha-r-bound} and \eqref{eq:b-N-bound}.
The second estimate is proved in a similar way.
\end{proof}
\begin{lem}
\label{lem:beta-series}
   For any $F \in \RR[X]/(X^N - 1)$ and any integer $\lambda \geq 0$ we have
   \[ \bignorm{\beta^*F - \sum_{r=0}^{\lambda-1} \beta^*_r F}
         \leq \frac87 \cdot 2^{-\lambda (b-1)} \norm{F}, \qquad
      \bignorm{\sum_{r=0}^{\lambda-1} \beta^*_r F}
         \leq \frac87 \norm{F}. \]
\end{lem}
\begin{proof}
   Follows from Lemma \ref{lem:beta-r-bound},
   similarly to the proof of Lemma \ref{lem:alpha-series}.
\end{proof}

Now we can establish that $\alpha^* F$ and $\beta^* F$
behave like the desired compositions $F(\alpha(X))$ and $F(\beta(X))$.
\begin{lem}
\label{lem:alpha-eval}
   Let $F \in \RR[X]/A(X)$, and let $z$ be a root of $X^N - 1$.
   Then
   \[ (\alpha^* F)(z) = F(\alpha(z)). \]
\end{lem}
\begin{rem}
   The expression $F(\alpha(z))$ is well-defined since $\alpha(z)$
   is a root of $A(X)$ (see Corollary \ref{cor:low-bijection}).
\end{rem}
\begin{proof}
   By the definition of $\alpha^*_r$,
   and since $z$ is a root of $X^N - 1$, we have
   \[ (\alpha^*_r F)(z) = \sum_{k=0}^{N-1} \alpha_{k,r} F_k z^{k+r}. \]
   Thus
   \[ (\alpha^* F)(z) = \sum_{r=0}^\infty (\alpha^*_r F)(z)
      = \sum_{k=0}^{N-1} F_k z^k \sum_{r=0}^\infty \alpha_{k,r} z^r
      = \sum_{k=0}^{N-1} F_k \alpha(z)^k = F(\alpha(z)). \qedhere \]
\end{proof}
\begin{lem}
\label{lem:beta-eval}
   Let $F \in \RR[X]/(X^N - 1)$, and let $z$ be a root of $A(X)$.
   Then
   \[ (\beta^* F)(z) = F(\beta(z)). \]
\end{lem}
\begin{proof}
   Similar to the proof of Lemma \ref{lem:alpha-eval}.
\end{proof}

\begin{cor}
\label{cor:alpha-beta-morphisms}
   The maps $\alpha^*$ and $\beta^*$ are mutually inverse ring isomorphisms
   between $\RR[X]/A(X)$ and $\RR[X]/(X^N - 1)$.
\end{cor}
\begin{proof}
   We have already pointed out that $\alpha^*$ and $\beta^*$ are linear;
   to show that they are ring homomorphisms we must show that they also
   respect multiplication.
   Lemma \ref{lem:alpha-eval} implies that
   for any $F, G \in \RR[X]/A(X)$ and any root $z$ of $X^N - 1$,
   we have
   \[ (\alpha^*(FG))(z) = (FG)(\alpha(z)) = F(\alpha(z)) \cdot G(\alpha(z))
      = (\alpha^* F)(z) \cdot (\alpha^* G)(z). \]
   Since a polynomial in $\RR[X]/(X^N - 1)$ is determined
   by its values at the roots of $X^N - 1$,
   this shows that $\alpha^*(FG) = (\alpha^* F)(\alpha^* G)$,
   and hence that $\alpha^*$ is a ring homomorphism.
   A similar argument using Lemma \ref{lem:beta-eval}
   shows that $\beta^*$ is a ring homomorphism.
   
   To show that $\alpha^*$ and $\beta^*$ are inverses,
   let $F \in \RR[X]/A(X)$ and let $z$ be a root of $A(X)$.
   Corollary \ref{cor:low-bijection} implies that
   \[ (\beta^* \alpha^* F)(z) = F(\alpha(\beta(z))) = F(z). \]
   Since this holds for all roots of $A(X)$,
   we see that $\beta^* \alpha^* F = F$.
   A similar argument shows that $\alpha^* \beta^* F = F$
   for all $F \in \RR[X]/(X^N - 1)$.
\end{proof}

Finally, we have the following two results concerning
the complexity of approximating $\alpha^*$ and $\beta^*$.

\begin{prop}[Approximating $\alpha^*$]
\label{prop:approx-alpha}
   Given as input $F \in 2^e \RR_p[X]/A(X)$,
   we may compute $G \in 2^{e+1} \RR_p[X]/(X^N - 1)$ such that
   \[ \norm{G - \alpha^* F} < 2^{e+1-p} \]
   in $O(N \Mint(p))$ bit operations, assuming that $p = O(b)$.
\end{prop}
Note that the output coefficients can indeed be represented in $2^{e+1} \RR_p$
thanks to the bound $\norm{\alpha^* F} \leq \frac{16}{15} \norm{F}$
(Lemma \ref{lem:alpha-series} with $\lambda = 0$).
A similar remark applies to Proposition \ref{prop:approx-beta} below
(via Lemma \ref{lem:beta-series}).
\begin{proof}
   Let $\lambda \coloneqq \lceil p/b \rceil$;
   the hypothesis $p = O(b)$ implies that $\lambda = O(1)$.
   According to Lemma~\ref{lem:alpha-series},
   \[ \bignorm{\sum_{r=0}^{\lambda-1} \alpha^*_r F}
   \leq \frac{16}{15} \norm{F} \leq 2^{e+1} \]
   and
   \[ \bignorm{\alpha^* F - \sum_{r=0}^{\lambda-1} \alpha^*_r F}
      \leq \frac{16}{15} \cdot 2^{-\lambda b} \norm{F}
      \leq \frac{16}{15} \cdot 2^{-p} 2^e
      = \frac{8}{15} \cdot 2^{e+1-p}. \]
   To compute the desired $G \in 2^{e+1} \RR_p[X]/(X^N - 1)$
   such that $\norm{G - \alpha^* F} < 2^{e+1-p}$,
   it suffices to ensure that $G$ satisfies
   \begin{equation}
   \label{eq:G-approx}
      \bignorm{G - \sum_{r=0}^{\lambda-1} \alpha^*_r F}
      < \frac{7}{15} \cdot 2^{e+1-p}.
   \end{equation}
   This may be accomplished by simply evaluating the sum
   $\sum_{r=0}^{\lambda-1} \alpha^*_r F$ directly from the definition,
   with a sufficiently high working precision.

   In more detail, we first calculate the coefficients $\alpha_{k,r}$,
   for each $r = 0, \ldots, \lambda-1$ and $k = 0, \ldots, N-1$.
   Each one requires $O(\lambda) = O(1)$ operations in $\RR$,
   using the usual formula for the binomial coefficients.
   Next we compute the coefficients of the polynomials
   \begin{align*}
   \alpha^*_0 F
      & = F_0 + F_1 X + F_2 X^2 + \cdots + F_{N-2} X^{N-2} + F_{N-1} X^{N-1}, \\
   \alpha^*_1 F
      & = (\alpha_{0,1} F_0 X + \alpha_{1,1} F_1 X^2 + \cdots +
            \alpha_{N-1,1} F_{N-1} X^N) \bmod X^N - 1 \\
      & = \alpha_{N-1,1} F_{N-1} + \alpha_{0,1} F_0 X + \cdots +
            \alpha_{N-2,1} F_{N-2} X^{N-1}, \\
   \alpha^*_2 F
      & = (\alpha_{0,2} F_0 X^2 + \alpha_{1,2} F_1 X^3 + \cdots +
            \alpha_{N-1,2} F_{N-1} X^{N+1}) \bmod X^N - 1 \\
      & = \alpha_{N-2,2} F_{N-2} + \alpha_{N-1,2} F_{N-1} X
            + \alpha_{0,2} F_0 X^2 + \cdots + \alpha_{N-3,2} F_{N-3} X^{N-1},
   \end{align*}
   and so on, up to $\alpha^*_{\lambda-1} F$.
   This costs altogether $O(\lambda N) = O(N)$ operations in $\RR$.
   Taking the sum of these polynomials costs another
   $O(\lambda N)$ operations in $\RR$.
   To ensure that \eqref{eq:G-approx} holds,
   it suffices to perform all of these operations with a working precision
   of $p + O(\log \lambda) = p + O(1)$ significant bits.
   The details of this error analysis are routine and are omitted.
   Each such addition, multiplication or division in $\RR$
   costs $O(\Mint(p))$ bit operations,
   leading to the claimed complexity bound.
\end{proof}

\begin{prop}[Approximating $\beta^*$]
\label{prop:approx-beta}
   Given as input $F \in 2^e \RR_p[X]/(X^N - 1)$,
   we may compute $G \in 2^{e+1} \RR_p[X]/A(X)$ such that
   \[ \norm{G - \beta^* F} < 2^{e+1-p} \]
   in $O(N \Mint(p))$ bit operations, assuming that $p = O(b)$.
\end{prop}
\begin{proof}
   Taking $\lambda \coloneqq \lceil p/(b-1) \rceil$,
   the proof proceeds along similar lines to that of
   Proposition~\ref{prop:approx-alpha},
   replacing the use of Lemma~\ref{lem:alpha-series}
   by Lemma~\ref{lem:beta-series}.
   The main difference is that the reductions modulo $A(X)$
   lead to slightly more complicated formulas.
   For example, we have
   \begin{align*}
      \beta^*_2 F
      & = (\beta_{0,2} F_0 X^2 + \cdots + \beta_{N-2,2} F_{N-2} X^N +
               \beta_{N-1,2} F_{N-1} X^{N+1}) \bmod A(X) \\
      & = \beta_{N-2,2} F_{N-2} +
         (\beta_{N-1,2} F_{N-1} - 2^{-b} \beta_{N-2,2} F_{N-2}) X \\
      & \phantom{abc} +
         (\beta_{0,2} F_0 - 2^{-b} \beta_{N-1,2} F_{N-1}) X^2 +
         \beta_{1,2} F_1 X^3 + \cdots
         + \beta_{N-3,2} F_{N-3} X^{N-1}.
   \end{align*}
   The terms with the minus signs are those arising
   from the $2^{-b} X$ term in $A(X)$.
   Overall, there are no more than $O(\lambda^2) = O(1)$ of these
   additional terms compared to the proof of
   Proposition~\ref{prop:approx-alpha}.
\end{proof}

\begin{rem}
   In the estimates given above,
   such as Lemma \ref{lem:alpha-beta-bound} and Lemma \ref{lem:alpha-r-bound},
   we have opted for shorter proofs rather than the sharpest possible bounds.
   With more effort, one could prove tighter bounds;
   this might save a few bits in the main algorithm,
   but does not affect the asymptotic conclusions of the paper.
   Similar remarks apply to the high product algorithm in
   Section \ref{sec:high}.
\end{rem}

\subsection{The main algorithm}

We are now in a position to state Algorithm \ref{algo:low}
and prove the main theorem concerning the computation of the low product.

\begin{algorithm}
   \SetAlgoLined
   \DontPrintSemicolon
   \KwIn{Parameters $n \geq 1$, $b \geq 4$, $N \geq 3$ with $N \geq n / b$,
   \newline integers $0 \leq u, v < 2^n$.}
   \KwOut{$uv \bmod 2^n$ (the low product of $u$ and $v$).}
   $p \assign 3b + \lg N + 6$.\;
   \stepcomment{Split inputs}
   Compute $u_0, \ldots, u_{N-1}$ and $v_0, \ldots, v_{N-1}$ with
   $0 \leq u_i, v_i < 2^b$ such that
   \[ u = \sum_{i=0}^{N-1} u_i 2^{ib},
      \qquad v = \sum_{i=0}^{N-1} v_i 2^{ib}, \]
   and let
   \[ \wbar U(X) \assign \sum_{i=0}^{N-1} u_i X^i,
      \qquad \wbar V(X) \assign \sum_{i=0}^{N-1} v_i X^i,
      \qquad \wbar U, \wbar V \in 2^b\RR_p[X]/A(X). \]
   \; \vspace{-10pt}
   \stepcomment{Convert to cyclic convolution}
   Use Proposition \ref{prop:approx-alpha} (approximating $\alpha^*$)
   to compute
   \[ \wtilde U, \wtilde V \in 2^{b+1}\RR_p[X]/(X^N - 1) \]
   such that
   \[ \norm{\wtilde U - \alpha^* \wbar U} < 2^{b+1-p},
      \qquad \norm{\wtilde V - \alpha^* \wbar V} < 2^{b+1-p}. \]
   \; \vspace{-10pt}
   \stepcomment{Perform convolution}
   Use \textsc{Convolution} (see \S\ref{sec:arithmetic}) to compute
   \[ \wtilde W \in 2^{2b+2+\lg N}\RR_p[X]/(X^N - 1) \]
   such that
   \[\norm{\wtilde W - \wtilde U \wtilde V} < 2^{2b+2+\lg N-p}. \]
   \; \vspace{-10pt}
   \stepcomment{Convert back}
   Use Proposition \ref{prop:approx-beta} (approximating $\beta^*$) to compute
   \[ \wbar W \in 2^{2b+3+\lg N}\RR_p[X]/A(X) \]
   such that
   \[ \norm{\wbar W - \beta^* \wtilde W} < 2^{2b+3+\lg N-p}. \]
   \; \vspace{-10pt}
   \stepcomment{Overlap-add}
   \KwRet{$\sum_{i=0}^{N-1} 2^{-b} \round(2^b \wbar {W_i})
         \cdot 2^{ib} \pmod{2^n}$}. \;
\caption{Low product}
\label{algo:low}
\end{algorithm}

\begin{proof}[Proof of Theorem \ref{thm:low}]
   As in the proof of Theorem \ref{thm:full}, let
   \[ U(X) \coloneqq \sum_{i=0}^{N-1} u_i X^i \in \ZZ[X], \qquad
      V(X) \coloneqq \sum_{i=0}^{N-1} v_i X^i \in \ZZ[X], \]
   so that $u = U(2^b)$ and $v = V(2^b)$, and let
   \[ W(X) \coloneqq U(X) V(X) = \sum_{i=0}^{2N-2} w_i X^i \in \ZZ[X]. \]
   The polynomials $\wbar U$ and $\wbar V$ in line 2 are just
   the images of $U$ and $V$ in $2^b \RR_p[X]/A(X)$.
   Our goal is to compute $L(X)$, the remainder on dividing $W(X)$ by $A(X)$,
   as in Proposition \ref{prop:low-cancel}.
   By definition this is equal to $\wbar U \wbar V$.

   Line 3 computes approximations $\wtilde U$ and $\wtilde V$
   to $\alpha^* \wbar U$ and $\alpha^* \wbar V$.
   Line 4 computes~$\wtilde W$,
   an approximation to $\wtilde U \wtilde V$
   (the cyclic convolution of $\wtilde U$ and $\wtilde V$).
   Observe that
   \begin{align*}
      \norm{\wtilde W - \alpha^*(\wbar U \wbar V)}
      & = \norm{\wtilde W - (\alpha^* \wbar U) (\alpha^* \wbar V)} \\
      & \leq \norm{\wtilde W - \wtilde U \wtilde V} +
         \norm{\wtilde U (\wtilde V - \alpha^* \wbar V)} +
         \norm{(\wtilde U - \alpha^* \wbar U)(\alpha^* \wbar V)} \\
      & \leq \norm{\wtilde W - \wtilde U \wtilde V} +
         N \norm{\wtilde U} \norm{\wtilde V - \alpha^* \wbar V} +
         N \norm{\alpha^* \wbar V}\norm{\wtilde U - \alpha^* \wbar U}.
   \end{align*}
   In this calculation we have used the fact that $\alpha^*$
   is a ring homomorphism (Corollary~\ref{cor:alpha-beta-morphisms}),
   and that $\norm{FG} \leq N \norm{F} \norm{G}$
   for any $F, G \in \RR[X]/(X^N - 1)$.
   By Lemma~\ref{lem:alpha-series} (with $\lambda = 0$) we have
   \[ \norm{\alpha^* \wbar V} \leq \frac{16}{15} \norm{\wbar V} < 2^{b+1}, \]
   so
   \begin{align*}
   \norm{\wtilde W - \alpha^*(\wbar U \wbar V)}
      & \leq 2^{2b+2+\lg N - p} + N \cdot 2^{b+1} 2^{b+1-p} +
                                  N \cdot 2^{b+1} 2^{b+1-p} \\
      & \leq 12 \cdot 2^{2b+\lg N-p}.
   \end{align*}
   
   Line 5 computes $\wbar W$, an approximation to $\beta^* \wtilde W$.
   Since $\alpha^*$ and $\beta^*$ are inverses
   (Corollary~\ref{cor:alpha-beta-morphisms}),
   Lemma \ref{lem:beta-series} implies that
   \begin{align*}
      \norm{\wbar W - L} = \norm{\wbar W - \wbar U \wbar V}
      & \leq \norm{\wbar W - \beta^* \wtilde W} +
         \norm{\beta^* (\wtilde W - \alpha^* (\wbar U \wbar V))} \\
      & < 2^{2b+3+\lg N-p} + \frac87 \cdot 12 \cdot 2^{2b+\lg N - p} \\
      & < 22 \cdot 2^{2b+\lg N - p}
         < 2^{-b}/2,
   \end{align*}
   where the last inequality follows from our choice of 
   $p = 3b + \lg N + 6$ in line 1.

   On the other hand, we know from Proposition \ref{prop:low-cancel}
   that $2^b L(X)$ has integer coefficients,
   so we deduce that $\round(2^b \wbar{W_i}) = 2^b L_i$ for each $i$.
   Therefore the sum in line~6 is equal to $L(2^b)$;
   by Proposition \ref{prop:low-cancel} this is equal to
   \begin{align*}
      \sum_{i=0}^{N-1} w_i 2^{ib}
         & \equiv \sum_{i=0}^{2N-2} w_i 2^{ib} \pmod{2^{Nb}} \\
         & = W(2^b) = uv.
   \end{align*}
   This congruence also holds modulo $2^n$ as $Nb \geq n$.

   The main term in the complexity bound \eqref{eq:low-bound} arises from
   the \textsc{Convolution} call in line 4.
   The splitting and overlap-add steps in lines 2 and 6 contribute
   $O(Nb)$ bit operations, as $\lg N = O(b)$ (by hypothesis),
   and the invocations of Proposition \ref{prop:approx-alpha} and
   Proposition~\ref{prop:approx-beta} in lines 3 and 5 contribute
   another $O(N \Mint(b))$ bit operations.
\end{proof}

\section{The high product}
\label{sec:high}

The discussion for the high product runs along similar lines to the low product,
with one additional technical complication.
The polynomial $B(X)$ that naturally replaces $A(X)$ in the cancellation trick
(see Proposition \ref{prop:high-cancel})
has $N$ roots near the roots of $X^N - 1$, just like $A(X)$,
but it also has a real root near $2^b$.
Some extra work is needed to handle this additional root.

\begin{rem}
   The asymmetry between the high and low products is somewhat mysterious.
   Perhaps it is related to the fact that in integer arithmetic,
   carries always propagate towards the most significant bits.
   The author has so far been unable to find a way of avoiding
   the annoying additional root.
\end{rem}

Throughout this section we continue to assume that \eqref{eq:b-N-bound} holds,
i.e., that $b \geq 4$ and $N \geq 3$.

\subsection{The cancellation trick}

We begin with a suitable analogue of Proposition~\ref{prop:low-cancel}.
To motivate our strategy,
recall that the cancellation trick for the low product relied on the fact that
\[ X^N \equiv 1 - 2^{-b} X \pmod{A(X)}. \]
Working modulo $A(X)$ has the effect of
shifting the high-order coefficients downwards by $N$ coefficients,
while at the same time multiplying them by $1 - 2^{-b} X$,
so that they will later cancel out when we evaluate at $X = 2^b$.
For the high product, we want to instead
multiply the \emph{low-order} coefficients by $1 - 2^{-b} X$
(to make them later cancel out),
and simultaneously shift the high-order coefficients
downward by $N$ coefficients.
We can accomplish these goals together by multiplying by $1 - 2^{-b} X$
modulo a polynomial $B(X)$ with the property that
\[ 1 - 2^{-b} X \equiv X^{-N} \pmod{B(X)}. \]
More precisely, we have the following result.

\begin{prop}
\label{prop:high-cancel}
   Let $W(X) \in \ZZ[X]$ with $\deg W \leq 2N$, say
   \[ W(X) = \sum_{i=0}^{2N} w_i X^i. \]
   Let $H(X) \in \RR[X]$ be the remainder on dividing $(1 - 2^{-b} X) W(X)$ by
   \[ B(X) \coloneqq X^{N+1} - 2^b X^N + 2^b \in \RR[X], \]
   with $\deg H < N + 1$.
   Then $2^b H(X) \in \ZZ[X]$ and
   \[ H(2^b) = \sum_{i=N}^{2N} w_i 2^{(i-N)b}. \]
\end{prop}
\begin{proof}
   Write
   \[ W(X) = \sum_{i=0}^{N-1} w_i X^i + \sum_{i=0}^{N} w_{N+i} X^{N+i}. \]
   Multiplying by $1 - 2^{-b} X$,
   and using the congruence $1 - 2^{-b} X \equiv X^{-N} \pmod{B(X)}$,
   we obtain
   \[ (1 - 2^{-b} X) W(X)
      \equiv \sum_{i=0}^{N-1} w_i X^i (1 - 2^{-b} X)
         + \sum_{i=0}^N w_{N+i} X^i \pmod{B(X)}. \]
   The polynomial on the right hand side has degree at most $N$,
   so we deduce that
   \[ H(X) = \sum_{i=0}^{N-1} w_i X^i (1 - 2^{-b} X)
            + \sum_{i=0}^N w_{N+i} X^i. \]
   This shows that $2^b H(X) \in \ZZ[X]$,
   and the result follows on substituting $X = 2^b$.
\end{proof}

\subsection{The roots of $B(X)$}
\label{sec:high-roots}

The next result isolates the auxiliary real root of $B(X)$.
\begin{lem}
\label{lem:real-root}
   The polynomial $B(X)$ has a unique real root $\rho$ in the interval
   \begin{equation}
   \label{eq:rho-interval}
      2^b (1 - 2 \cdot 2^{-Nb}) < \rho < 2^b.
   \end{equation}
   In particular,
   \begin{equation}
   \label{eq:rho-interval-numeric}
      0.999 \cdot 2^b < \rho < 2^b.
   \end{equation}
\end{lem}

\begin{rem}
   It turns out that $\rho$ is very close to $2^b(1 - 2^{-Nb})$,
   i.e., the midpoint of the interval \eqref{eq:rho-interval}.
   In fact, one can develop a series expansion for $\rho$,
   whose first few terms are given by
   \[ \rho = 2^b (1 - 2^{-Nb} - N \cdot 2^{-2Nb} - \cdots), \]
   but we will not prove this.
\end{rem}

\begin{proof}
   It is convenient to make the transformation
   \[ P(Y) \coloneqq \frac{B(2^b Y)}{2^{(N+1)b}} = Y^{N+1} - Y^N + \epsilon
   \qquad \text{where } \epsilon \coloneqq 2^{-Nb}. \]
   Our goal is to show that $P(Y)$ has a unique real root in the interval
   $(1-2\epsilon, 1)$.
   
   Clearly $P(1) = \epsilon > 0$.
   We claim that $P(1-2\epsilon) < 0$.
   To see this, observe that
   \[ P(1-2\epsilon) = (1-2\epsilon)^{N+1} - (1-2\epsilon)^N + \epsilon
      = \epsilon(1 - 2(1-2\epsilon)^N). \]
   From \eqref{eq:b-N-bound} we have $\epsilon = 2^{-Nb} \leq 2^{-4N}$ and then
   $(1 - 2\epsilon)^N \geq (1 - 2^{-4N+1})^N > \frac12$,
   so indeed $P(1-2\epsilon) < 0$.
   By the intermediate value theorem,
   $P(Y)$ has at least one root in $(1-2\epsilon,1)$.
   
   To prove that there is exactly one root,
   we will show that $P'(Y) > 0$ throughout the interval.
   We have $P'(Y) = Y^{N-1} ((N+1) Y - N)$,
   so it suffices to show that $1 - 2\epsilon > N/(N+1)$,
   i.e., that $2\epsilon < 1/(N+1)$.
   This is clear as $2\epsilon = 2^{-Nb+1} \leq 2^{-4N+1}$.
   
   Finally, \eqref{eq:rho-interval-numeric} follows immediately
   from \eqref{eq:rho-interval},
   after taking into account \eqref{eq:b-N-bound}.
\end{proof}

We note for future use the identity
\begin{equation}
\label{eq:rho-identity}
   1 - 2^{-b} \rho = \rho^{-N},
\end{equation}
which follows immediately from the fact that $B(\rho) = 0$.

Next consider the polynomial
\[ C(X) \coloneqq \frac{B(X)}{X - \rho} \in \RR[X]. \]
The coefficients of $C(X)$ are given explicitly as follows.
\begin{lem}
   We have
   \begin{equation}
      \label{eq:C-explicit}
      C(X) = X^N - \frac{2^b}{\rho}
         \left( \frac{X^{N-1}}{\rho^{N-1}} + \cdots + \frac{X}{\rho} + 1 \right).
   \end{equation}
\end{lem}
\begin{proof}
   First observe that
   \[ C(X) = \frac{X^{N+1} - \rho X^N + \rho X^N - 2^b X^N + 2^b}{X - \rho}
      = X^N + \frac{(\rho - 2^b) X^N + 2^b}{X - \rho}. \]
   From \eqref{eq:rho-identity} we have $\rho - 2^b = -2^b/\rho^N$ and hence
   \[ C(X) = X^N - \frac{2^b}{\rho} \cdot \frac{(X/\rho)^N - 1}{X/\rho - 1}.
      \qedhere \]
\end{proof}

\begin{lem}
\label{lem:high-roots}
   The roots of $C(X)$ lie in $D_2$, and they are all simple.
\end{lem}
\begin{proof}
   If $z$ is a root of $C(X)$ and $|z| \geq 1$, then
   \[ |z|^N = \frac{2^b}{\rho}
         \left| \frac{z^{N-1}}{\rho^{N-1}} + \cdots + \frac{z}{\rho} + 1 \right|
      \leq \frac{2^b}{\rho} |z|^{N-1}
         \left( \frac{1}{\rho^{N-1}} + \cdots + \frac{1}{\rho} + 1\right), \]
   so by \eqref{eq:rho-interval-numeric} and \eqref{eq:b-N-bound} we obtain
   \[ |z| \leq \frac{2^b}{\rho} \cdot \frac{1}{1 - \rho^{-1}}
   = \frac{2^b}{\rho - 1} \leq \frac{2^b}{0.999 \cdot 2^b - 1} < 2. \]

   If $C(X)$ had a multiple root, say $z$,
   then $z$ would also be a multiple root of $B(X)$.
   This would imply that
   \[ B'(z) = (N+1)z^N - 2^b N z^{N-1} = 0, \]
   which in turn forces $z = 2^b N / (N+1)$
   (since clearly $z \neq 0$).
   This contradicts the previous paragraph,
   as $2^b N/(N+1)$ does not lie in $D_2$.
\end{proof}

Lemma \ref{lem:real-root} and Lemma \ref{lem:high-roots} together imply
that $B(X)$ has $N+1$ distinct roots, namely,
the $N$ roots of $C(X)$, and the auxiliary root $\rho$.
Figure \ref{fig:high-roots} illustrates the case $N = 12$, $b = 1$.
\begin{figure}
\begin{tikzpicture}[scale=2]
   \draw (-1.2, 0) -- (2.2, 0);
   \draw (0, -1.2) -- (0, 1.2);
   \draw (-1.000, 0.000) circle(0.02);
   \draw (1.000, 0.000) circle(0.02);
   \draw (-0.866, -0.500) circle(0.02);
   \draw (-0.866, 0.500) circle(0.02);
   \draw (-0.500, -0.866) circle(0.02);
   \draw (-0.500, 0.866) circle(0.02);
   \draw (0.000, -1.000) circle(0.02);
   \draw (0.000, 1.000) circle(0.02);
   \draw (0.500, -0.866) circle(0.02);
   \draw (0.500, 0.866) circle(0.02);
   \draw (0.866, -0.500) circle(0.02);
   \draw (0.866, 0.500) circle(0.02);
   \filldraw[black] (-0.968, 0.000) circle(0.02);
   \filldraw[black] (1.065, 0.000) circle(0.02);
   \filldraw[black] (2.000, 0.000) circle(0.02);
   \filldraw[black] (-0.847, -0.473) circle(0.02);
   \filldraw[black] (-0.847, 0.473) circle(0.02);
   \filldraw[black] (-0.511, -0.833) circle(0.02);
   \filldraw[black] (-0.511, 0.833) circle(0.02);
   \filldraw[black] (-0.037, -0.989) circle(0.02);
   \filldraw[black] (-0.037, 0.989) circle(0.02);
   \filldraw[black] (0.466, -0.896) circle(0.02);
   \filldraw[black] (0.466, 0.896) circle(0.02);
   \filldraw[black] (0.880, -0.554) circle(0.02);
   \filldraw[black] (0.880, 0.554) circle(0.02);
\end{tikzpicture}
\caption{Roots of $X^{13} - 2X^{12} + 2$ (filled circles)
      and $X^{12} - 1$ (empty circles).}
\label{fig:high-roots}
\end{figure}

Now consider the function
\[ \delta(z) \coloneqq z(1 - 2^{-b}z)^{1/N}, \qquad z \in D_{2^b}. \]
This is the same as the definition of $\beta(z)$
in Section \ref{sec:low-roots},
except that the exponent $-1/N$ has been replaced by $1/N$.
The roots of $C(X)$ lie well within the domain of definition of $\delta(z)$.
The auxiliary root $\rho$ is also inside the domain,
but lies very close to the boundary.

\begin{lem}
The function $\delta(z)$ maps roots of $B(X)$ to roots of $X^N - 1$.
\end{lem}
\begin{proof}
If $z$ is a root of $B(X)$, then
 \[ \delta(z)^N = z^N(1 - 2^{-b} z) = z^N - 2^{-b} z^{N+1} = 1. \qedhere \]
\end{proof}
Of course, $\delta(z)$ cannot yield a bijection
between the roots of $B(X)$ and those of $X^N - 1$,
as $B(X)$ has too many roots.
In a moment we will see that we do get a bijection
if we restrict to the roots of $C(X)$.
(It turns out that $\delta(\rho) = 1$,
so $\delta$ maps precisely two roots of $B(X)$ to $1$.
We omit the easy proof.)

For any $k \in \ZZ$, the function $\delta(z)^k$ is represented
on $D_{2^b}$ by the series
\[ \delta(z)^k = z^k \sum_{r=0}^\infty \delta_{k,r} z^r
   = z^k + \delta_{k,1} z^{k+1} + \delta_{k,2} z^{k+2} + \cdots \]
where
 \[ \delta_{k,r} \coloneqq \binom{k/N}{r} (-2^{-b})^r, \qquad r \geq 0. \]
Again, $\delta_{k,r}$ is identical to $\beta_{k,r}$,
except that $N$ has the opposite sign.
The first few terms in the expansion of $\delta(z)$ are
\[ \delta(z) = z - \frac{1}N 2^{-b} z^2 - \frac{(N-1)}{2N^2} 2^{-2b} z^3
   - \frac{(N-1)(2N-1)}{6N^3} 2^{-3b} z^4 - \cdots. \]

Let $\gamma(z) \in z \, \RR[[z]]$ be the formal series inverse of $\delta(z)$.
\begin{lem}
\label{lem:gamma-coeffs}
   For any $k \geq 0$ we have (formally)
   \[ \gamma(z)^k = z^k \sum_{r=0}^\infty \gamma_{k,r} z^r
      = z^k +  \gamma_{k,1} z^{k+1} + \gamma_{k,2} z^{k+2} + \cdots \]
   where $\gamma_{k,0} \coloneqq 1$ and
   \[ \gamma_{k,r} \coloneqq \frac{k}{k+r} \binom{-(k+r)/N}{r} (-2^{-b})^r,
      \qquad r \geq 1. \]
\end{lem}
In particular, the first few terms of $\gamma(z)$ are
\[ \gamma(z) = z + \frac{1}N 2^{-b} z^2 + \frac{(N+3)}{2 N^2} 2^{-2b} z^3
   + \frac{(N+4)(2N+4)}{6 N^3} 2^{-3b} z^4 + \cdots. \]
\begin{proof}
   Same as the proof of Lemma \ref{lem:alpha-coeffs},
   with $N$ replaced by $-N$ everywhere.
\end{proof}

\begin{lem}
\label{lem:gamma-delta-bound}
   For all $r \geq 0$ and $0 \leq k < N$ we have
   \[ |\delta_{k,r}| \leq 2^{-rb}, \qquad |\gamma_{k,r}| \leq 2^{-r(b-2)}. \]
\end{lem}
\begin{proof}
   The bound for $\delta_{k,r}$ follows by the same argument used for
   $\beta_{k,r}$ in the proof of Lemma \ref{lem:alpha-beta-bound}.
   For $\gamma_{k,r}$, observe that for $r \geq 1$ we have
   \begin{multline*}
      \frac{|\gamma_{k,r}|}{2^{-rb}}
      = \frac{k}{k+r} \cdot \frac{1}{r!} \prod_{j=0}^{r-1} \left(\frac{k+r}{N} + j\right)
      \leq \frac{1}{r!} \prod_{j=0}^{r-1} \left( \frac{r}{N} + j + 1 \right) \\
      \leq \frac{1}{r!} \left(\frac{r}{N} + r\right)^r
      = \frac{r^r}{r!} \left(1 + \frac{1}{N}\right)^r.
   \end{multline*}
   Stirling's formula implies that $r^r/r! \leq e^r$,
   so since $N \geq 3$ we obtain
   \[ \frac{|\gamma_{k,r}|}{2^{-rb}} \leq e^r (4/3)^r < 3.63^r < 2^{2r}. \qedhere \]
\end{proof}

\begin{rem}
   The constant $-2$ in the above bound for $\gamma_{k,r}$
   ensures that the statement is correct for all $k$ and $r$,
   but asymptotically speaking it is not really necessary.
   In fact one can prove that for any $\epsilon > 0$,
   there exist $N_0$ and $r_0$ such that
   $|\gamma_{k,r}| \leq 2^{-r(b-\epsilon)}$ for all $0 \leq k < N$,
   whenever $N \geq N_0$ and $r \geq r_0$.
\end{rem}

\begin{cor}
\label{cor:high-inverses}
   The series for $\gamma(z)$ and $\delta(z)$ converge on
   $D_{2^{b-2}}$ and $D_{2^b}$ respectively, and
   \[ \gamma(\delta(z)) = z = \delta(\gamma(z)), \qquad z \in D_2. \]
\end{cor}
\begin{proof}
   Same as the proof of Corollary \ref{cor:low-inverses},
   first using Lemma \ref{lem:gamma-delta-bound} to show that
   $\delta(z)$ maps $D_2$ into $D_3 \subseteq D_{2^{b-2}}$ and that
   $\gamma(z)$ maps $D_2$ into $D_4 \subseteq D_{2^b}$.
\end{proof}

\begin{cor}
\label{cor:high-bijection}
   The functions $\gamma(z)$ and $\delta(z)$ induce mutually inverse bijections
   between the roots of $X^N - 1$ and the roots of $C(X)$.
\end{cor}
\begin{proof}
   Similar to the proof of Corollary \ref{cor:low-bijection}.
\end{proof}

\subsection{Ring isomorphisms}
\label{sec:high-iso}

In this section we will first construct maps
\begin{align*}
   \gamma^* & \colon \RR[X]/C(X) \longto \RR[X]/(X^N - 1), \\
   \delta^* & \colon \RR[X]/(X^N - 1) \longto \RR[X]/C(X),
\end{align*}
analogous to the maps $\alpha^*$ and $\beta^*$ defined
in Section \ref{sec:low-iso}.
Note that these maps do not yet take into account the auxiliary root $\rho$.

For each $r \geq 0$ define linear maps
\begin{align*}
   \gamma^*_r & \colon \RR[X]/C(X) \longto \RR[X]/(X^N - 1), \\
   \delta^*_r & \colon \RR[X]/(X^N - 1) \longto \RR[X]/C(X)
\end{align*}
by the formulas
\begin{align*}
   \gamma^*_r \bigg(\sum_{k=0}^{N-1} F_k X^k \bmod{C(X)}\bigg)
      & \coloneqq \sum_{k=0}^{N-1} \gamma_{k,r} F_k X^{k+r} \bmod{X^N - 1}, \\
   \delta^*_r \bigg(\sum_{k=0}^{N-1} F_k X^k \bmod{X^N - 1}\bigg)
      & \coloneqq \sum_{k=0}^{N-1} \delta_{k,r} F_k X^{k+r} \bmod{C(X)}.
\end{align*}
As in Section \ref{sec:low-iso} we have the following norm bounds.
\begin{lem}
\label{lem:gamma-r-bound}
   For any $r \geq 0$ and $F \in \RR[X]/C(X)$,
   \[ \norm{\gamma^*_r F} \leq 2^{-r(b-2)} \norm{F}. \]
\end{lem}
\begin{proof}
   Similar to the proof of Lemma \ref{lem:alpha-r-bound},
   using Lemma \ref{lem:gamma-delta-bound} to bound the series coefficients.
\end{proof}
\begin{lem}
\label{lem:delta-r-bound}
   For any $r \geq 0$ and $F \in \RR[X]/(X^N - 1)$,
   \[ \norm{\delta^*_r F} \leq 2^{-r(b-1)} \norm{F}. \]
\end{lem}
\begin{proof}
   As in the proof of Lemma \ref{lem:beta-r-bound},
   we must first work out the effect of multiplication by $X$ modulo $C(X)$.
   Let $H = \sum_{k=0}^{N-1} H_k X^k \in \RR[X]/C(X)$.
   The formula \eqref{eq:C-explicit} implies that
   \[ XH = \frac{2^b}{\rho} H_{N-1}
      + \left( H_0 + \frac{2^b}{\rho^2} H_{N-1} \right) X + \cdots
      + \left(H_{N-2} + \frac{2^b}{\rho^N} H_{N-1} \right) X^{N-1}. \]
   We have $2^b/\rho < 2$ and $2^b/\rho^i < 1$ for $i \geq 2$
   (due to \eqref{eq:rho-interval-numeric} and \eqref{eq:b-N-bound}),
   so we find that $\norm{XH} \leq 2 \norm{H}$.

   The rest of the argument is the same as
   the proof of Lemma \ref{lem:beta-r-bound},
   noting that $\delta_r^* F = X^r G$ for
   $G \coloneqq \sum_{k=0}^{N-1} \delta_{k,r} F_k X^k \in \RR[X]/C(X)$,
   and using Lemma \ref{lem:gamma-delta-bound}.
\end{proof}

We now define $\gamma^*$ and $\delta^*$ by setting
\[ \gamma^* F \coloneqq \sum_{r=0}^\infty \gamma^*_r F,
   \qquad \delta^* F \coloneqq \sum_{r=0}^\infty \delta^*_r F. \]
The next five statements are proved along the same lines as the
corresponding results in Section \ref{sec:low-iso},
i.e., from Lemma \ref{lem:alpha-series} up to
Corollary \ref{cor:alpha-beta-morphisms}.
\begin{lem}
\label{lem:gamma-series}
   For any $F \in \RR[X]/C(X)$ and any integer $\lambda \geq 0$, we have
   \[ \bignorm{\gamma^* F - \sum_{r=0}^{\lambda-1} \gamma^*_r F}
         \leq \frac43 \cdot 2^{-\lambda (b-2)} \norm{F}, \qquad
      \bignorm{\sum_{r=0}^{\lambda-1} \gamma^*_r F}
         \leq \frac43 \norm{F}. \]
\end{lem}
\begin{lem}
\label{lem:delta-series}
   For any $F \in \RR[X]/(X^N - 1)$ and any integer $\lambda \geq 0$, we have
   \[ \bignorm{\delta^* F - \sum_{r=0}^{\lambda-1} \delta^*_r F}
         \leq \frac87 \cdot 2^{-\lambda (b-1)} \norm{F}, \qquad
      \bignorm{\sum_{r=0}^{\lambda-1} \delta^*_r F}
         \leq \frac87 \norm{F}. \]
\end{lem}

\begin{lem}
\label{lem:gamma-eval}
   Let $F \in \RR[X]/C(X)$, and let $z$ be a root of $X^N - 1$.
   Then
   \[ (\gamma^* F)(z) = F(\gamma(z)). \]
\end{lem}
\begin{lem}
\label{lem:delta-eval}
   Let $F \in \RR[X]/(X^N - 1)$, and let $z$ be a root of $C(X)$.
   Then
   \[ (\delta^* F)(z) = F(\delta(z)) \]
\end{lem}

\begin{cor}
\label{cor:gamma-delta-morphisms}
   The maps $\gamma^*$ and $\delta^*$ are mutually inverse ring isomorphisms
   between $\RR[X]/C(X)$ and $\RR[X]/(X^N - 1)$.
\end{cor}

Now we bring $\rho$ back into the picture.
We will define maps
\begin{align*}
   \gamma^\dagger & \colon \RR[X]/B(X) \longto \RR[X]/(X^N - 1) \oplus \RR, \\
   \delta^\dagger & \colon \RR[X]/(X^N - 1) \oplus \RR \longto \RR[X]/B(X),
\end{align*}
in terms of $\gamma^*$ and $\delta^*$, as follows.

First, for $F \in \RR[X]/B(X)$, we define
\[ \gamma^\dagger F \coloneqq
   \big(\gamma^*(F \bmod C(X)), \rho^{-N} F(\rho)\big)
   \in \RR[X]/(X^N - 1) \oplus \RR. \]
The map $\gamma^\dagger$ is a linear isomorphism,
thanks to the Chinese remainder theorem
applied to the relatively prime moduli $C(X)$ and $X - \rho$.
However, $\gamma^\dagger$ is not quite a ring isomorphism,
i.e., is not multiplicative,
due to the scaling factor $\rho^{-N}$.
Note that the second component of $\gamma^\dagger F$ may be written
more explicitly as follows:
if $F = F_0 + F_1 X + \cdots + F_N X^N$, then
\[ \rho^{-N} F(\rho) = F_N + \rho^{-1} F_{N-1} + \cdots + \rho^{-N} F_0. \]
   
In the other direction, for
\[ (F, \theta) \in \RR[X]/(X^N - 1) \oplus \RR, \]
we define $\delta^\dagger(F, \theta)$ to be the
unique polynomial $G \in \RR[X]/B(X)$ such that
\[ G \equiv (1 - 2^{-b} X) \delta^*(F) \pmod{C(X)},
   \qquad  G(\rho) = \rho^N \theta. \]
Again, $\delta^\dagger$ is a linear isomorphism, but not a ring isomorphism.

In fact, $\gamma^\dagger$ and $\delta^\dagger$ are not even
inverse to each other.
Instead, they have been cooked up to satisfy the following relation.
\begin{lem}
\label{lem:delta-multiply}
   For any $F, G \in \RR[X]/B(X)$ we have
   \[ \delta^\dagger(\gamma^\dagger F \cdot \gamma^\dagger G)
      = (1 - 2^{-b} X) FG. \]
\end{lem}
(The reason for including the factor $1 - 2^{-b} X$ is to prepare for the use
of Proposition \ref{prop:high-cancel} in the main high product algorithm.)
\begin{proof}
   It is enough to check the equality modulo $C(X)$ and modulo $X - \rho$.
   It holds modulo $C(X)$ according to
   Corollary \ref{cor:gamma-delta-morphisms}:
   \begin{align*}
      \delta^\dagger(\gamma^\dagger F \cdot \gamma^\dagger G)
         & \equiv (1 - 2^{-b} X)
            \delta^*\big(\gamma^*(F \bmod C(X)) \gamma^*(G \bmod C(X))\big) \\
         & = (1 - 2^{-b} X) \delta^*\big(\gamma^*(FG \bmod C(X))\big) \\
         & = (1 - 2^{-b} X) FG \pmod{C(X)}.
   \end{align*}
   It holds modulo $X - \rho$ thanks to \eqref{eq:rho-identity}:
   \begin{align*}
      \delta^\dagger(\gamma^\dagger F \cdot \gamma^\dagger G)(\rho)
         & = \rho^N (\rho^{-N} F(\rho))(\rho^{-N} G(\rho)) \\
         & = \rho^{-N} F(\rho) G(\rho) = (1 - 2^{-b} \rho) (FG)(\rho) \qedhere.
   \end{align*}
\end{proof}

Define a norm on $\RR[X]/(X^N - 1) \oplus \RR$ by taking
\begin{align*}
   \norm{(F, \theta)} & \coloneqq \max(\norm{F}, |\theta|) \\
   & = \max(|F_0|, \ldots, |F_{N-1}|, |\theta|).
\end{align*}
Then $\gamma^\dagger$ and $\delta^\dagger$ satisfy the following norm bounds.
\begin{lem}
\label{lem:gamma-dagger-norm}
   For any $F \in \RR[X]/B(X)$ we have
   \[ \norm{\gamma^\dagger F} \leq 3 \norm{F}. \]
\end{lem}
\begin{proof}
   Let $F = F_0 + \cdots + F_N X^N \in \RR[X]/B(X)$.
   Using \eqref{eq:C-explicit},
   we find that the reduction of $F$ modulo $C(X)$ is given by
   \begin{equation}
   \label{eq:reduction-modC}
      \left(F_0 + \frac{2^b}{\rho} F_N\right) +
         \left(F_1 + \frac{2^b}{\rho^2} F_N\right) X + \cdots
         + \left(F_{N-1} + \frac{2^b}{\rho^N} F_N\right) X^{N-1},
   \end{equation}
   so by \eqref{eq:rho-interval-numeric} we obtain
   \[ \norm{F \bmod C(X)} \leq \left(1 + \frac{2^b}{\rho}\right) \norm{F}
      \leq 2.002 \norm{F}. \]
   Lemma \ref{lem:gamma-series} (with $\lambda = 0$) then yields
   \[ \norm{\gamma^*(F \bmod C(X))}
      \leq \frac43 \cdot 2.002 \norm{F} \leq 3 \norm{F}. \]
   We also have
   \[ |\rho^{-N} F(\rho)| = |F_N + \rho^{-1} F_{N-1} + \cdots + \rho^{-N} F_0 |
      \leq \frac{1}{1 - \rho^{-1}} \norm{F} \leq 2 \norm{F}. \]
   Together these inequalities show that
   $\norm{\gamma^\dagger F} \leq 3 \norm{F}$.
\end{proof}
\begin{lem}
\label{lem:delta-dagger-norm}
   For any $(F, \theta) \in \RR[X]/(X^N - 1) \oplus \RR$ we have
   \[ \norm{\delta^\dagger (F, \theta)} \leq 3 \norm{(F, \theta)}. \]
\end{lem}
\begin{proof}
   We may write down an explicit formula for
   $\delta^\dagger(F, \theta)$ as follows.
   Let
   \[ H \coloneqq \delta^* F \in \RR[X]/C(X), \]
   let $\wbar H \in \RR[X]$ be the unique polynomial of degree less than $N$
   such that $\wbar H \equiv H \pmod{C(X)}$,
   and let
   \begin{equation}
   \label{eq:psi-formula}
      \psi \coloneqq
         \frac{\rho^N \theta - \rho^{-N} \wbar H(\rho)}{C(\rho)} \in\RR.
   \end{equation}
   Then we claim that
   \begin{equation}
   \label{eq:crt-sum}
      \delta^\dagger(F, \theta) \equiv
         (1 - 2^{-b} X) \wbar H(X) + \psi \, C(X) \pmod{B(X)}.
   \end{equation}
   
   To prove \eqref{eq:crt-sum},
   it suffices to verify that it holds modulo $C(X)$ and modulo $X - \rho$.
   For $C(X)$ this is clear as
   $\delta^\dagger(F,\theta) \equiv (1 - 2^{-b} X) \delta^* F \pmod{C(X)}$
   by the definition of $\delta^\dagger$.
   It holds modulo $X - \rho$ because by \eqref{eq:rho-identity}
   and the definition of $\delta^\dagger$ we have
   \begin{multline*}
      (1 - 2^{-b} \rho) \wbar H(\rho) + \psi \, C(\rho)
         = \rho^{-N} \wbar H(\rho) +
            \frac{\rho^N \theta - \rho^{-N} \wbar H(\rho)}{C(\rho)}
            \cdot C(\rho) \\
         = \rho^N \theta = \delta^\dagger(F,\theta)(\rho).
   \end{multline*}
   
   Our goal is now to estimate the size of the coefficients of the polynomial
   \[ J(X) \coloneqq (1 - 2^{-b} X) \wbar H(X) + \psi \, C(X) \in \RR[X]. \]
   Note that $J(X)$ has degree at most~$N$,
   so its coefficients are exactly the same
   as those of $\delta^\dagger(F,\theta)$.
   
   Let us first estimate $|\psi|$.
   Write $H(X) = H_0 + \cdots + H_{N-1} X^{N-1}$,
   so that also $\wbar H(X) = H_0 + \cdots + H_{N-1} X^{N-1}$.
   We have
   \begin{multline*}
      |\rho^{-N} \wbar H(\rho)|
      = |H_{N-1} \rho^{-1} + \cdots + H_0 \rho^{-N}| \\
      \leq \frac{\rho^{-1}}{1 - \rho^{-1}} \max(|H_0|, \ldots, |H_{N-1}|)
      = \frac{1}{\rho - 1} \norm{H}.
   \end{multline*}
   By Lemma \ref{lem:delta-series} (with $\lambda = 0$) we have
   $\norm{H} = \norm{\delta^* F} \leq \frac87\norm{F}$.
   Using \eqref{eq:rho-interval-numeric} and \eqref{eq:b-N-bound} we obtain
   \[ |\rho^{-N} \wbar H(\rho)|
      \leq \frac87 \cdot \frac{1}{0.999 \cdot 2^b - 1} \norm{F}
      \leq \norm{F}. \]
   From \eqref{eq:C-explicit} and \eqref{eq:rho-interval-numeric} we have
   \[ C(\rho) = \rho^N - \frac{2^b}{\rho} N \geq \rho^N - 1.002 N. \]
   Therefore
   \begin{multline*}
   |\psi| \leq \frac{\rho^N |\theta| + |\rho^{-N} \wbar H(\rho)|}{C(\rho)}
      \leq \frac{\rho^N |\theta| + \norm{F}}{\rho^N - 1.002 N}
      = \frac{|\theta| + \rho^{-N} \norm{F}}{1 - 1.002 N / \rho^N} \\
      \leq \frac{1 + \rho^{-N}}{1 - 1.002 N / \rho^N}
         \max(\norm{F}, |\theta|).
   \end{multline*}
   Since $N \geq 3$ and $\rho > 15$
   (again from \eqref{eq:rho-interval-numeric} and \eqref{eq:b-N-bound}),
   we find that
   \[ \frac{1 + \rho^{-N}}{1 - 1.002 N/\rho^N} < 1.002, \]
   so
   \[ |\psi| \leq 1.002 \norm{(F,\theta)}. \]

   Now we may estimate the size of the coefficients of $J(X)$.
   The coefficients of $(1 - 2^{-b} X) \wbar H(X)$ are bounded
   in absolute value by $(1 + 2^{-b}) \norm{H}$,
   and from \eqref{eq:C-explicit} we see that the coefficients of
   $\psi \, C(X)$ are bounded in absolute value by $(2^b/\rho) |\psi|$.
   Therefore, using again
   \eqref{eq:b-N-bound} and \eqref{eq:rho-interval-numeric}
   we find that the coefficients of $J(X)$ are bounded in absolute value by
   \[ (1 + 2^{-b}) \norm{H} + \frac{2^b}{\rho} |\psi|
      \leq \frac{17}{16} \cdot \frac{8}{7} \norm{F} + 1.002^2 \norm{(F,\theta)}
      \leq 3 \norm{(F,\theta)}. \qedhere \]
\end{proof}

Next, we exhibit efficient algorithms for approximating
$\gamma^\dagger$ and $\delta^\dagger$.

\begin{prop}[Approximating $\gamma^\dagger$]
\label{prop:approx-gamma}
   Given as input $F \in 2^e \RR_p[X]/B(X)$, we may compute
   \[ G \in 2^{e+2} \RR_p[X]/(X^N - 1), \qquad \theta \in 2^{e+2} \RR_p, \]
   such that
   \[ \norm{(G,\theta) - \gamma^\dagger F} < 2^{e+2-p}, \]
   in $O(N \Mint(p))$ bit operations, assuming that $p = O(b)$.
\end{prop}
\begin{proof}
   We first remark that since $p = O(b)$,
   we may precompute $\rho$ to a precision of~$p + O(1)$ significant bits
   using only $O(\log N)$ operations in $\RR$,
   by using Newton's method to numerically solve the equation $B(z) = 0$,
   starting with the initial approximation $z = 2^b$.
   (See also Remark \ref{rem:newton} below.)
   
   Now, given as input $F \in 2^e \RR_p[X]/B(X)$ as above,
   we first compute an approximation to $F \bmod C(X)$ using the formula
   \eqref{eq:reduction-modC}.
   The hypothesis $p = O(b)$,
   together with the rapid decay of the coefficients of $C(X)$,
   implies that this may be done using $O(1)$ operations in $\RR$.
   We may then compute the desired approximation~$G$ to
   $\gamma^*(F \bmod C(X))$ using the same method as in
   the proof of Proposition \ref{prop:approx-alpha},
   at a cost of $O(N)$ operations in $\RR$.
   Finally, we may easily compute the desired approximation $\theta$ to
   $\rho^{-N} F(\rho) = F_N + \rho^{-1} F_{N-1} + \cdots$
   using another $O(1)$ operations in $\RR$.
\end{proof}

\begin{prop}[Approximating $\delta^\dagger$]
\label{prop:approx-delta}
   Given as input $F \in 2^e \RR_p[X]/(X^N - 1)$ and $\theta \in 2^e \RR_p$,
   we may compute
   \[ G \in 2^{e+2} \RR_p[X]/B(X) \]
   such that
   \[ \norm{G - \delta^\dagger(F, \theta)} < 2^{e+2-p} \]
   in $O(N \Mint(p))$ bit operations, assuming that $p = O(b)$.
\end{prop}
\begin{proof}
   The algorithm amounts to evaluating the explicit formula \eqref{eq:crt-sum}.
   We first approximate $H = \delta^* F$ using the same method
   as in the proof of Proposition \ref{prop:approx-beta},
   at a cost of $O(N)$ operations in $\RR$.
   (This requires $O(1)$ more operations than
   the corresponding algorithm for $\beta^*$,
   because the reductions modulo $C(X)$ involve a few more terms
   than those modulo $A(X)$.)
   We then approximate $\psi$ at a cost of $O(1)$ operations,
   and evaluate \eqref{eq:crt-sum} in another $O(N)$ operations.
\end{proof}

\begin{rem}
\label{rem:newton}
   In practice we always have $Nb \gg p$,
   and this assumption allows several simplifications to be made to the
   algorithms in Proposition \ref{prop:approx-gamma} and
   Proposition \ref{prop:approx-delta}.
   First, the trivial approximation $\rho \approx 2^b$ is already correct
   to $p + O(1)$ significant bits, so Newton's method is not required.
   In addition, we have $C(\rho) \approx \rho^N \approx 2^{Nb}$,
   so instead of the complicated formula \eqref{eq:psi-formula} for $\psi$,
   we may simply use the approximation $\psi \approx \theta$.
\end{rem}

Finally we may state the main high product algorithm, and prove the main theorem concerning its correctness and complexity.

\begin{algorithm}
   \SetAlgoLined
   \DontPrintSemicolon
   \KwIn{Parameters $n \geq 1$, $b \geq 4$, $N \geq 3$
      with $(N + 1)b \geq n + \lg N + 2$,
      \newline integers $0 \leq u, v < 2^n$.}
   \KwOut{$0 \leq w \leq 2^n$ such that $|uv - 2^n w| < 2^n$
      \newline (the high product of $u$ and $v$).}
   $p \assign 3b + \lg N + 9$.\;
   \stepcomment{Split inputs}
   Compute $u_0, \ldots, u_N$ and $v_0, \ldots, v_N$
   with $0 \leq u_i, v_i < 2^b$ such that
   \[ u = \sum_{i=0}^N u_i 2^{ib - ((N+1)b-n)},
      \qquad v = \sum_{i=0}^N v_i 2^{ib - ((N+1)b-n)}, \]
   and let
   \[ \wbar U(X) \assign \sum_{i=0}^N u_i X^i,
         \qquad \wbar V(X) \assign \sum_{i=0}^N v_i X^i,
         \qquad \wbar U, \wbar V \in 2^b\RR_p[X]/B(X). \]
   \; \vspace{-10pt}
   \stepcomment{Convert to cyclic convolution}
   Use Proposition \ref{prop:approx-gamma} (approximating $\gamma^\dagger$)
   to compute
   \[ \wtilde U, \wtilde V \in 2^{b+2}\RR_p[X]/(X^N - 1),
      \qquad \theta_U, \theta_V \in 2^{b+2} \RR_p, \]
   such that
   \[ \norm{(\wtilde U, \theta_U) - \gamma^\dagger \wbar U} < 2^{b+2-p}, \qquad
      \norm{(\wtilde V, \theta_V) - \gamma^\dagger \wbar V} < 2^{b+2-p}. \]
   \; \vspace{-10pt}
   \stepcomment{Perform convolution}
   Use \textsc{Convolution} (see \S\ref{sec:arithmetic}) to compute
   \[ \wtilde W \in 2^{2b+4+\lg N}\RR_p[X]/(X^N - 1),
         \qquad \theta_W \in 2^{2b+4+\lg N} \RR_p \]
   such that
   \[\norm{\wtilde W - \wtilde U \wtilde V} < 2^{2b+4+\lg N-p},
      \qquad |\theta_W - \theta_U \theta_V| < 2^{2b+4+\lg N-p}. \]
   \; \vspace{-10pt}
   \stepcomment{Convert back}
   Use Proposition \ref{prop:approx-delta} (approximating $\delta^\dagger$)
   to compute
   \[ \wbar W \in 2^{2b+6+\lg N}\RR_p[X]/B(X) \]
   such that
   \[ \norm{\wbar W - \delta^\dagger(\wtilde W, \theta_W)}
      < 2^{2b+6+\lg N-p}. \]
   \; \vspace{-10pt}
   \stepcomment{Overlap-add}
   $t \coloneqq \sum_{i=0}^N
      2^{-b} \round(2^b \wbar W_i) \cdot 2^{ib - ((N+2)b-n)}$.\;
   \KwRet{$\round(t)$}.\;
\caption{High product}
\label{algo:high}
\end{algorithm}

\begin{proof}[Proof of Theorem \ref{thm:high}]
   Line 2 decomposes $u$ and $v$ into $N+1$ chunks of $b$ bits.
   The splitting boundaries are different to those used
   for the full product and low product:
   here $u_N$ consists of the $b$ most significant bits of $u$,
   then $u_{N-1}$ the next lower $b$ bits, and so on.
   The hypothesis $N + 1 \geq n/b$ ensures that this splitting is possible.

   As in the proof of Theorem \ref{thm:full}, let
   \[ U(X) \coloneqq \sum_{i=0}^N u_i X^i \in \ZZ[X],
      \qquad V(X) \coloneqq \sum_{i=0}^N v_i X^i \in \ZZ[X], \]
   so that
   \[ u = \frac{U(2^b)}{2^{(N+1)b-n}},
      \qquad v = \frac{V(2^b)}{2^{(N+1)b-n}}. \]
   Let
   \[ W(X) \coloneqq U(X) V(X) = \sum_{i=0}^{2N} w_i X^i \in \ZZ[X]. \]
   The polynomials $\wbar U$ and $\wbar V$ in line 2 are just
   the images of $U$ and $V$ in $2^b \RR_p[X]/B(X)$.
   Our goal is to compute $H(X)$,
   the remainder on dividing $(1 - 2^{-b} X) W(X)$ by $B(X)$,
   as in Proposition~\ref{prop:high-cancel}.
   By definition this is equal to
   $(1 - 2^{-b} X) \wbar U \wbar V \pmod{B(X)}$.

   Line 3 computes approximations
   $(\wtilde U, \theta_U)$ and $(\wtilde V, \theta_V)$ to
   $\gamma^\dagger \wbar U$ and $\gamma^\dagger \wbar V$.
   Line 4 computes $\wtilde W$, an approximation to $\wtilde U \wtilde V$,
   and $\theta_W$, an approximation to $\theta_U \theta_V \in \RR$.
   The latter involves just a single real multiplication.
   Let us write $\wbar U'$ and $\wbar V'$ for the images of
   $\wbar U$ and $\wbar V$ in $\RR[X]/C(X)$.
   A similar calculation to that used in
   the proof of Theorem \ref{thm:low} shows that
   \begin{align*}
      \norm{\wtilde W - (\gamma^* \wbar U')(\gamma^* \wbar V')}
      & \leq \norm{\wtilde W - \wtilde U \wtilde V}
         + \norm{\wtilde U (\wtilde V - \gamma^* \wbar V')}
         + \norm{(\gamma^* \wbar V') (\wtilde U - \gamma^* \wbar U')} \\
      & \leq \norm{\wtilde W - \wtilde U \wtilde V}
         + N \norm{\wtilde U} \norm{\wtilde V - \gamma^* \wbar V'}
         + N \norm{\gamma^* \wbar V'} \norm{\wtilde U - \gamma^* \wbar U'} \\
      & \leq 2^{2b+4+\lg N-p} + N \cdot 2^{b+2} 2^{b+2-p}
         + N \cdot 2^{b+2} 2^{b+2-p} \\
      & \leq 48 \cdot 2^{2b+\lg N-p}
   \end{align*}
   and that
   \begin{align*}
   |\theta_W - (\rho^{-N} \wbar U(\rho)) (\rho^{-N} \wbar V(\rho))|
   & \leq |\theta_W - \theta_U \theta_V| + |\theta_U| |\theta_V - \rho^{-N} \wbar V(\rho)| \\
   & \phantom{abcdefghijkl} + |\rho^{-N} \wbar V(\rho)| |\theta_U - \rho^{-N} \wbar U(\rho)| \\
   & \leq 2^{2b+4+\lg N-p} + 2^{b+2} 2^{b+2-p} + 2^{b+2} 2^{b+2-p} \\
   & \leq 48 \cdot 2^{2b+\lg N-p}.
   \end{align*}
   These two inequalities may be expressed more briefly in combination
   by writing
   \[ \norm{(\wtilde W, \theta_W) -
         (\gamma^\dagger \wbar U)(\gamma^\dagger \wbar V)}
         \leq 48 \cdot 2^{2b+\lg N-p}. \]

   Line 5 computes an approximation $\wbar W$ to
   $\delta^\dagger (\wtilde W, \theta_W)$.
   Using Lemma \ref{lem:delta-multiply} and
   Lemma \ref{lem:delta-dagger-norm} we find that
   \begin{align*}
      \norm{\wbar W - (1 - 2^{-b} X) \wbar U \wbar V}
      & = \norm{\wbar W - \delta^\dagger(
         \gamma^\dagger \wbar U \cdot \gamma^\dagger \wbar V)} \\
      & \leq \norm{\wbar W - \delta^\dagger(\wtilde W, \theta_W)}
         + \norm{\delta^\dagger((\wtilde W, \theta_W)
            - \gamma^\dagger \wbar U \cdot \gamma^\dagger \wbar V)} \\
      & \leq 2^{2b+6+\lg N-p} + 3 \cdot 48 \cdot 2^{2b+\lg N - p} \\
      & = 208 \cdot 2^{2b+\lg N-p} < 2^{-b}/2.
   \end{align*}

   As noted earlier, the coefficients of $H(X)$ are exactly those of
   $(1 - 2^{-b} X) \wbar U \wbar V$.
   We know from Proposition \ref{prop:high-cancel} that
   $2^b H(X)$ has integer coefficients,
   so we deduce that $\round(2^b \wbar W_i) = 2^b H_i$ for each
   $i = 0, \ldots, N$.
   Thus the value $t$ computed in line~6 satisfies
   \[ t = \frac{H(2^b)}{2^{(N+2)b-n}}. \]
   Applying Proposition \ref{prop:high-cancel}, we obtain
   
   \begin{multline*}
      uv - 2^n t
         = \frac{U(2^b)V(2^b)}{2^{(2N+2)b-2n}} - \frac{H(2^b)}{2^{(N+2)b-2n}} \\
         = \frac{W(2^b) - 2^{Nb} H(2^b)}{2^{(2N+2)b-2n}}
         = \frac{1}{2^{(2N+2)b-2n}} \sum_{i=0}^{N-1} w_i 2^{ib}.
   \end{multline*}
   We have $w_i \leq 2^{2b} N$ for $i = 0, \ldots, N-1$,
   since each $w_i$ is a sum of $i+1$ terms of the form $u_j v_k$,
   so
   \[ 0 \leq uv - 2^n t \leq \frac{1}{2^{(2N+2)b-2n}} \cdot (2^{2b} N)
         \cdot \frac{2^{bN} - 1}{2^b - 1}
         \leq \frac{16}{15} \cdot \frac{N}{2^{(N+1)b - 2n}}.  \]
   The hypothesis $(N + 1)b \geq n + \lg N + 2$ then yields
   \begin{equation}
   \label{eq:uv-error}
      0 \leq uv - 2^n t < \frac{2^n}{2}. 
   \end{equation}
   Let $w \coloneqq \round(t)$ be the value returned in line 7.
   Since $0 \leq uv < 2^{2n}$, the inequality \eqref{eq:uv-error}
   implies that $-\frac12 < t < 2^n$, and hence that $0 \leq w \leq 2^n$.
   Moreover, since $|t - w| \leq \frac12$, we conclude that
   \[ |uv - 2^n w| \leq |uv - 2^n t| + 2^n |t - w|
      < \frac{2^n}{2} + \frac{2^n}{2} = 2^n \]
   as desired.
   The running time analysis is essentially the same as in
   the proof of Theorem \ref{thm:low}.
\end{proof}

\section{Implementation and performance}
\label{sec:code}

We wrote an implementation of the new truncated product algorithms
in the C programming language,
together with a comparable implementation of the full product,
to examine to what extent the predicted 25\% reduction in complexity
can be realised in practice.
The source code is available from the author's web page
under a free software license.

\begin{table}[h]
   \caption{Timings for full and truncated products.
   Values in parentheses indicate ratio of times for
   truncated vs full product.}
   \begin{tabular}{r|cc|cc|c|c}
      \toprule
      \multicolumn{1}{c|}{$n$} &
      \multicolumn{2}{c|}{low product} &
      \multicolumn{2}{c|}{high product} & full product & GMP \\
      \midrule
            1\,000\,000 & 2.90ms & (0.94) & 2.93ms & (0.95) & 3.07ms & 2.68ms \\
            2\,154\,434 & 7.11ms & (1.02) & 7.27ms & (1.05) & 6.95ms & 6.93ms \\
            4\,641\,588 & 14.7ms & (0.95) & 15.2ms & (0.99) & 15.4ms & 16.6ms \\
           10\,000\,000 & 36.2ms & (0.92) & 37.8ms & (0.96) & 39.5ms & 39.1ms \\
           21\,544\,346 & 88.6ms & (0.89) & 92.6ms & (0.93) & 99.2ms & 99.4ms \\
           46\,415\,888 &  204ms & (0.90) &  210ms & (0.93) &  227ms &  237ms \\
          100\,000\,000 &  504ms & (0.84) &  514ms & (0.86) &  598ms &  553ms \\
          215\,443\,469 &  1.25s & (0.91) &  1.28s & (0.93) &  1.37s &  1.35s \\
          464\,158\,883 &  2.76s & (0.91) &  2.81s & (0.93) &  3.03s &  3.05s \\
       1\,000\,000\,000 &  6.08s & (0.86) &  6.19s & (0.88) &  7.05s &  6.93s \\
       2\,154\,434\,690 &  13.9s & (0.86) &  14.2s & (0.88) &  16.1s &  17.0s \\
       4\,641\,588\,833 &  33.6s & (1.01) &  34.6s & (1.04) &  33.4s &  38.1s \\
      10\,000\,000\,000 &   109s & (1.37) &   110s & (1.38) &  79.8s &  81.6s \\
      \bottomrule
   \end{tabular}
   \label{tab:timings}
\end{table}

The timings reported in Table \ref{tab:timings} were run on a
single core of an otherwise idle 2.5GHz Intel Xeon Gold 6248
(Cascade Lake microarchitecture),
running Rocky Linux 8.8 (kernel version 4.18.0).
We compiled our program using GCC 12.2.0 with the optimisation flags
\texttt{-O3 -mavx2 -mavx512f -ffast-math}.
In the critical inner loops, our code uses GCC's vector extensions to take advantage of the AVX2 instruction set available on the target platform.

For the real convolutions,
our code relies on the one-dimensional
real-to-complex and complex-to-real transforms
provided by the FFTW library (version 3.3.10) \cite{Fri-fftw}.
We configured FFTW using the \verb!--enable-avx2 --enable-avx512! flags,
and used FFTW's ``wisdom'' facility with the \verb!FFTW_MEASURE! option to find
efficient transform sequences for all relevant transform lengths.

Our implementation differs from the theoretical presentation in
Section \ref{sec:low} and Section \ref{sec:high} in several respects:
\begin{itemize}
   \item
   Instead of fixed point arithmetic, we use double-precision floating point
   (the \texttt{double} data type in C).
   In particular, this applies to the routines that compute
   $\alpha^*$, $\beta^*$, $\gamma^\dagger$ and $\delta^\dagger$,
   and also the FFTs and pointwise multiplications.
   (The splitting and overlap-add steps are handled using integer arithmetic.)
   We make no attempt to prove any bounds for round-off error.
   This is impossible anyway in the context of our program,
   as FFTW does not offer any error guarantees.
   
   \item
   In the splitting step we allow signed coefficients.
   For example, we write $u = U(2^b)$ where the coefficients of $U$ are integers lying in the balanced interval $|U_i| \leq 2^{b-1}$.
   This leads to less coefficient growth in the product $U(X) V(X)$: instead of these coefficients having roughly $2b + \lg N$ bits, for uniformly random inputs they tend to have around $2b + \frac12 \lg N$ bits, due to cancellation between the positive and negative terms.
   Of course, an adversary could easily choose inputs for which every $U_i$ and $V_i$ is close to $2^{b-1}$, in which case the product coefficients will have close to $2b + \lg N$ bits.
   In this case our program will certainly produce incorrect output, unless we decrease $b$ to compensate.
\end{itemize}

Table \ref{tab:timings} shows timings for our low product, high product,
and full product routines for various choices of $n$
(with parameters as indicated in Table \ref{tab:parameters}),
as well as timings for the full product computed by the \verb|mpz_mul| function
from the GMP multiple-precision arithmetic library
(version 6.2.1) \cite{gmp-6.2.1}.
The reported timings are averages taken over numerous tests;
for each entry in the table, the measurements are quite stable,
with standard deviation around 1--2\% of the average time.

Comparing the performance of our code against GMP is not quite fair,
because in principle GMP performs a provably correct computation,
whereas the output of our program is not provably correct
(as explained above).
Nevertheless, the timings demonstrate that our code is competitive with the
highly optimised multiplication routines in GMP.

For each $n$ shown in Table \ref{tab:timings},
we chose the parameters as follows.
We ran a large number of tests to determine the maximum possible $b$ for which
the program consistently produces the correct output
for uniformly random inputs $u$ and~$v$.
The parameter $\lambda$ refers to the number of terms used in
the approximation of the ring isomorphisms such as $\alpha^*$;
it has the same meaning as in the proof of Proposition~\ref{prop:approx-alpha}.
Again, we chose $\lambda$ by empirical testing,
taking the smallest value that led to consistently correct output.
Regarding the choice of $N$,
we examined several possible candidates,
namely those of the form $N = 2^{e_2} 3^{e_3} 5^{e_5} 7^{e_7}$
where $n/b \leq N \leq 1.15 n/b$ and $3^{e_3} 5^{e_5} 7^{e_7} < 200$,
and chose the candidate that led to the fastest timings.
(In principle the choice of $N$ could also affect correctness,
due to different FFT algorithms being used for different $N$,
but in practice we found it did not make a difference.)
The resulting values of $N$, $b$ and $\lambda$ are shown
in Table \ref{tab:parameters};
these are the values that were used to produce
the timings in Table \ref{tab:timings}.

The final column in Table~\ref{tab:parameters} gives the ratio between
the transform length used for the truncated products
(the column labelled $N$)
and the corresponding transform length for the full product
(the column labelled $2N$).
As discussed in \S\ref{sec:results} (scenario \#2),
we expect this ratio to be close to $3/4$.
The observed values are reasonably close to $3/4$,
but there is some variation due to the sparsity of available transform lengths.
We also predicted in \S\ref{sec:results} that the ratio of the corresponding
values of $b$ should be about $2/3$;
this is borne out clearly in Table \ref{tab:parameters}.

\newcommand{\slimcdot}{\mspace{1.5mu}\mathord{\cdot}\mspace{1.5mu}}

\begin{table}[h]
\caption{Parameter choices}
\centerline{
\begin{tabular}{r|r@{\hspace{3pt}}lc@{\hspace{6pt}}c|r@{\hspace{3pt}}lc|c}
\toprule
& \multicolumn{4}{c|}{truncated products}
& \multicolumn{3}{c|}{full product} & length \\
\multicolumn{1}{c|}{$n$} &
$N$ & & $b$ & $\lambda$ & $2N$ & & $b$ & ratio \\
\midrule
1\,000\,000 &
71\,680 & $= 2^{11} \slimcdot 35$ & 14 & 4 &
98\,304 & $= 2^{15} \slimcdot 3$ & 21 & 0.73 \\
2\,154\,434 &
165\,888 & $= 2^{11} \slimcdot 81$ & 13 & 4 &
215\,040 & $= 2^{11} \slimcdot 105$ & 21 & 0.77 \\
4\,641\,588 &
358\,400 & $= 2^{11} \slimcdot 175$ & 13 & 4 &
512\,000 & $= 2^{12} \slimcdot 125$ & 20 & 0.70 \\
10\,000\,000 &
786\,432 & $= 2^{18} \slimcdot 3$ & 13 & 4 &
1\,024\,000 & $= 2^{13} \slimcdot 125$ & 20 & 0.77 \\
21\,544\,346 &
1\,769\,472 & $= 2^{16} \slimcdot 27$ & 13 & 4 &
2\,211\,840 & $= 2^{14} \slimcdot 135$ & 20 & 0.80 \\
46\,415\,888 &
3\,670\,016 & $= 2^{19} \slimcdot 7$ & 13 & 4 &
4\,915\,200 & $= 2^{16} \slimcdot 75$ & 19 & 0.75 \\
100\,000\,000 &
8\,388\,608 & $= 2^{23}$ & 12 & 4 &
10\,616\,832 & $= 2^{17} \slimcdot 81$ & 19 & 0.79 \\
215\,443\,469 &
18\,874\,368 & $= 2^{21} \slimcdot 9$ & 12 & 4 &
23\,592\,960 & $= 2^{19} \slimcdot 45$ & 19 & 0.80 \\
464\,158\,883 &
39\,321\,600 & $= 2^{19} \slimcdot 75$ & 12 & 5 &
52\,428\,800 & $= 2^{21} \slimcdot 25$ & 19 & 0.75 \\
1\,000\,000\,000 &
83\,886\,080 & $= 2^{24} \slimcdot 5$ & 12 & 5 &
113\,246\,208 & $= 2^{22} \slimcdot 27$ & 18 & 0.74 \\
2\,154\,434\,690 &
188\,743\,680 & $= 2^{22} \slimcdot 45$ & 12 & 5 &
262\,144\,000 & $= 2^{21} \slimcdot 125$ & 18 & 0.72 \\
4\,641\,588\,833 &
452\,984\,832 & $= 2^{24} \slimcdot 27$ & 11 & 5 &
524\,288\,000 & $= 2^{22} \slimcdot 125$ & 18 & 0.86 \\
10\,000\,000\,000 &
939\,524\,096 & $= 2^{27} \slimcdot 7$ & 11 & 5 &
1\,207\,959\,552 & $= 2^{27} \slimcdot 9$ & 17 & 0.78 \\
\bottomrule
\end{tabular} }
\label{tab:parameters}
\end{table}

Finally, let us discuss the main quantity of interest,
namely, the ratio of the running times between the truncated products
and the full product.
These ratios are shown in parentheses in Table \ref{tab:timings}.
In an ideal world these ratios would be close to $0.75$.
Unfortunately, the values shown in the table fall somewhat short of this goal.
As $n$ increases,
the performance appears to pass through three distinct phases:
\begin{itemize}
   \item
   In the first few rows of the table, the ratio is close to $1$.
   For these small values of $n$,
   the savings from the shorter transform lengths in the truncated products
   are outweighed by the additional cost of evaluating the ring isomorphisms.
   \item
   There is a wide range of intermediate values of $n$
   where the ratios for the low product lie roughly between 0.85 and 0.9,
   i.e., we see a speedup of 10--15\% compared to the full product.
   The high product lags behind by about 2--3\%.
   \item
   In the last two rows of the table,
   there is a serious degradation in performance.
   This appears to be mainly due to problems with FFTW's handling of
   large transforms of composite length.

   For example, for the low product in the last line,
   the FFTs of size $2^{27} \cdot 7$ account for roughly 87s of the total 109s,
   whereas for the corresponding full product,
   the FFTs of length $2^{27} \cdot 9$ account
   for about 63s out of the total 80s.
   One would normally expect the ratio of these FFT times to be about
   $7/9 \approx 0.78$ rather than $87/63 \approx 1.38$.
   We did not explore the underlying reasons for this discrepancy,
   but we speculate that it is caused by suboptimal locality in the
   algorithms that FFTW uses to decompose a large composite-length FFT
   into smaller transforms.
\end{itemize}

In summary, our implementation shows that it is possible to achieve
a nontrivial speedup for the computation of truncated products,
over a wide range of values of $n$.
The author is hopeful that the running times may be improved further
by more careful optimisation work,
especially in the subroutines for computing the ring isomorphisms.
It also seems likely that the performance for large values of $n$
may be improved by suitable tweaking of the underlying FFT algorithms.

\section*{Acknowledgments}

The author thanks Joris van der Hoeven for his comments
on a draft of this paper.
The performance tests were carried out on
the Katana computing cluster at UNSW \cite{katana};
many thanks to Martin Thompson for technical support.
The author was supported by the Australian Research Council,
grants DP150101689 and FT160100219.

\bibliographystyle{amsplain}
\bibliography{truncmul}

\end{document}